\documentclass[sigconf]{acmart}

\usepackage{microtype}      
\usepackage{xcolor}         
\usepackage{amsmath}
\usepackage{setspace}
\usepackage{graphicx} 
\usepackage{float} 
\usepackage{subfigure} 
\usepackage{multirow}
\usepackage{diagbox}
\usepackage{color}
\usepackage{enumitem}
\usepackage{booktabs}
\usepackage{subfigure} 
\usepackage{cleveref}
\usepackage{tabularx}
\usepackage{ulem}

\Crefformat{equation}{Equation~#2#1#3}
\usepackage{makecell}

\usepackage{booktabs}
\usepackage{amsthm}
\newtheorem{theorem}{Theorem}[section]

\AtBeginDocument{%
  }

\copyrightyear{2025}
\acmYear{2025}
\setcopyright{acmlicensed}\acmConference[MM '25]{Proceedings of the 33rd
ACM International Conference on Multimedia}{October 27--31, 2025}{Dublin,
Ireland}
\acmBooktitle{Proceedings of the 33rd ACM International Conference on
Multimedia (MM '25), October 27--31, 2025, Dublin, Ireland}
\acmDOI{10.1145/3746027.3754564}
\acmISBN{979-8-4007-2035-2/2025/10}

\begin{document}

\title{Why Generate When You Can Transform? Unleashing Generative Attention for Dynamic Recommendation}

\author{Yuli Liu}
\affiliation{%
  \institution{Quan Cheng Laboratory, Jinan \\
  Qinghai University, Xining}
  \country{China}
  }
\email{liuyuli012@gmail.com}

\author{Wenjun Kong}
\affiliation{%
  \institution{School of Computer Technology and Application \\ Qinghai University, Xining}
  \country{China}
  }
\email{wenjunkong6@gmail.com}

\author{Cheng Luo}
\affiliation{%
  \institution{Quan Cheng Laboratory, Jinan \\
  MegaTech.AI, Beijing}
  \country{China}
  }
\email{luochengleo@gmail.com}

\author{Weizhi Ma}
\authornote{Corresponding author.}
\affiliation{%
  \institution{AIR, Tsinghua University, Beijing \\ 
  Quan Cheng Laboratory, Jinan}
  \country{China}
  }
\email{mawz@tsinghua.edu.cn}


\renewcommand{\shortauthors}{Liu et al.}

\begin{abstract}
Sequential Recommendation (SR) focuses on personalizing user experiences by predicting future preferences based on historical interactions. Transformer models, with their attention mechanisms, have become the dominant architecture in SR tasks due to their ability to capture dependencies in user behavior sequences. However, traditional attention mechanisms, where attention weights are computed through query-key transformations, are inherently linear and deterministic. This fixed approach limits their ability to account for the dynamic and non-linear nature of user preferences, leading to challenges in capturing evolving interests and subtle behavioral patterns. Given that generative models excel at capturing non-linearity and probabilistic variability, we argue that generating attention distributions offers a more flexible and expressive alternative compared to traditional attention mechanisms. To support this claim, we present a theoretical proof demonstrating that generative attention mechanisms offer greater expressiveness and stochasticity than traditional deterministic approaches. Building upon this theoretical foundation, we introduce two generative attention models for SR, each grounded in the principles of Variational Autoencoders (VAE) and Diffusion Models (DMs), respectively. These models are designed specifically to generate adaptive attention distributions that better align with variable user preferences. Extensive experiments on real-world datasets show our models significantly outperform state-of-the-art in both accuracy and diversity.
\end{abstract}




\maketitle

\vspace{-0.6mm}
\section{INTRODUCTION}
\label{sec:intro}
\vspace{-0.8mm}

Sequential Recommendation (SR) is a critical task in many modern applications \cite{ liu2016context, wang2023sequential} with the goal of predicting the next item a user may interact with based on their historical sequence of interactions. 
Common techniques for SR include the earlier Markov models \cite{he2016fusing, cai2017spmc}, matrix factorization-based methods \cite{ khan2025convseq, zheng2020modeling}, convolutional neural networks (CNNs) \cite{zhou2020cnn, tang2018personalized, yan2019cosrec}, recurrent neural networks (RNNs) \cite{xu2019recurrent, cui2018mv, liu2025pone}, graph neural networks (GNNs) \cite{yu2022element, liu2024selfgnn, zhang2022dynamic}, and more recently, Transformer-based models \cite{ li2025dgt, chen2025dualcfgl, zivic2024scaling}. Among these, Transformer architectures \cite{vaswani2017attention}, with their attention mechanisms \cite{ sun2019bert4rec, kang2018self, ge2025personalized}, have become the dominant architecture in SR due to their ability to effectively capture dependencies within long and complex sequences of user behavior.

However, the uncertainty in user behavior and the complexity of behavioral patterns \cite{cheng2025sequential, han2024efficient}, along with the nature of SR tasks, \textit{i.e.}, dynamic and evolving user preferences \cite{chen2021modeling, boka2024survey}, presents distinct challenges.
Traditional attention mechanisms, which primarily rely on query-key transformations, compute attention scores linearly through a dot product \cite{tang2021probabilistic, noci2024shaped}. This process is deterministic, meaning that for a given set of queries and keys, the resulting attention weights are static and consistently calculated \cite{lin2022structured, martin2020monte}.
This fixed approach limits the model's expressiveness for capturing complex patterns and stochasticity for adapting to dynamic user preferences. 
As a result, traditional attention mechanisms' ability to adapt to real-world SR environments is reduced.
Although some works have attempted to improve the expressiveness of attention mechanisms by modifying the computation form \cite{liu2024probabilistic, tian2024eulerformer, zhou2024theoretical} or introducing probabilistic representations to incorporate stochasticity \cite{fan2022sequential, martin2020monte, tang2021probabilistic}, they still remain largely dependent on linear transformations and attention scores with a relatively fixed computation formula.
Consequently, Transformer-based SR models still have not undergone a fundamental change and continue to face challenges in integrating stochasticity while enhancing their expressiveness and adaptability.

Given these challenges, the inherent advantages of generative models (\textit{e.g.}, their ability to handle uncertainty and capture complex, non-linear dependencies \cite{yoon2017semi, bond2021deep}) highlight their potential as a promising alternative for overcoming the limitations of existing attention mechanisms. 
Unlike deterministic linear transformations, generative models can learn to represent intricate patterns and uncertainties directly, enabling more adaptive and expressive computations.
Building upon these advantages, we propose a novel perspective: leveraging \underline{Gen}erative models to directly generate \underline{Att}ention weight distributions (GenAtt) for SR.
This perspective fundamentally shifts away from the reliance on traditional fixed computation formulas and static representations, opening up new possibilities for more flexible and expressive framework that addresses the limitations of existing Transformer-based SR models.

To advance the GenAtt perspective, we first provide a theoretical demonstration highlighting the advantages of generative attention distributions over traditional deterministic attention mechanisms, particularly in terms of their ability to integrate stochasticity and enhance expressiveness. This theoretical foundation establishes that generative models offer richer, more dynamic representations of user behavior, effectively addressing the inherent variability and uncertainty in SR tasks.
Building on this foundation, we propose two distinct generative attention models tailored to sequential recommendation, each leveraging the unique strengths of Variational Autoencoders (VAEs) \cite{kingma2013auto} and Diffusion Models (DMs) \cite{ho2020denoising}. 
VAEs and DMs are selected due to their widespread use in generative tasks and their capacity to model latent variables in a probabilistic manner, which aligns well with the objective of dynamically learning attention distributions.
The VAE-based model learns compact probabilistic representations to address uncertainty and variability in user behavior \cite{correia2023continuous, liang2018variational}, enabling it to generalize across diverse interactions and reveal latent patterns. 
The DM-based model leverages its iterative refinement process to generate adaptive attention distributions \cite{wu2019neural, zhao2024denoising, wang2023better}. Its ability to progressively model complex, non-linear dependencies and inherent noise makes it  suited for capturing fine-grained temporal dynamics.
The integration of theoretical analysis with experimental validation ensures a comprehensive exploration of GenAtt mechanisms. Our approach not only demonstrates the theoretical potential of generative attention in improving expressiveness and stochasticity but also validates its effectiveness in real-world SR scenarios. 
By modeling stochastic latent representations and generating adaptive attention distributions, GenAtt improves recommendation relevance while simultaneously enhancing diversity. The key contributions of our study include:
\vspace{-0.6mm}
\begin{itemize}[leftmargin=*]
\vspace{-0.5mm}
    \item We approach sequential recommendation from a generative perspective, utilizing generative models to directly generate attention weight distributions, which uses an unsupervised latent distribution learning process to replace traditional fixed transformations. By thoroughly exploring this perspective, it demonstrates how a more flexible and adaptive attention mechanism can effectively capture complex behavior patterns and dynamic user preferences, as opposed to relying on static or deterministic computation methods.
    
    \item We provide a theoretical proof that demonstrates the advantages of generative attention distributions, emphasizing their enhanced expressiveness and ability to model uncertainty. This theoretical foundation supports the effectiveness of generative models in improving the performance of sequential recommendation.
    \item We propose two models, each based on distinct generative mechanisms, \textit{i.e.}, VAEs and DMs. These models introduce entirely new forms of attention mechanisms, moving away from fixed computation formulas and linear transformations typically used in traditional approaches. By seamlessly integrating the unique strengths of VAEs and DMs with Transformer, both expressiveness and robustness of GenAtt-based models are enhanced.
    \item Our extensive experiments confirm the theoretical insights, showing that our approaches not only improve recommendation relevance but also promote diversity. 
\end{itemize}

\vspace{-1mm}
\section{RELATED WORK}
\vspace{-0.5mm}
\label{sec:related-work}
\textbf{Sequential Recommendation Models.}
Early approaches are based on Markov models \cite{he2016fusing, cai2017spmc}, which assumed that the next action only depends on a limited number of previous actions \cite{xu2024online}, making them effective but overly simplistic for capturing long-term dependencies. With the advent of deep learning, models such as CNNs are introduced \cite{zhou2020cnn, yan2019cosrec}, and the seminal model \cite{tang2018personalized} leverages their ability to extract local patterns within interaction sequences. Recurrent neural networks \cite{graves2012long} further improve SR by capturing temporal dependencies over longer sequences \cite{xu2019recurrent, cui2018mv, liu2025pone}, offering a more robust understanding of user behavior. GNNs extend these capabilities by modeling the relationships between items and users \cite{yu2022element, liu2024selfgnn, zhang2022dynamic}, enabling richer representations and often being combined with sequential models to enhance recommendation performance. In recent years, Transformer architectures have revolutionized SR \cite{sun2019bert4rec, ge2025personalized}, with their ability to model long-range dependencies and complex user behavior patterns through self-attention mechanisms \cite{kang2018self, li2020time, liu2024pay, liu2024probabilistic}. \\
\textbf{Generative Models for SR.}
Unlike discriminative models \cite{chen2022elecrec, liu2022cdarl}, which focus on directly learning mappings between input features and target labels, generative models aim to model the underlying data distribution \cite{deldjoo2024review}, enabling them to generate new samples and capture uncertainty in a probabilistic manner. In the context of SR, generative models are often employed to address challenges such as modeling inherent uncertainties in user behavior and enhancing the diversity of recommendations. For instance, IRGAN \cite{wang2017irgan} and SRecGAN \cite{lu2021srecgan} utilize adversarial training to generate user-item interactions and next items, respectively. 
VAE-based approaches \cite{xie2021adversarial, shenbin2020recvae} leverage latent probabilistic representations to better capture user preferences. 
Diffusion models have also emerged as promising tools for refining representations \cite{li2023diffurec, du2023sequential} or generate discrete items \cite{wu2019neural, jiang2024diffkg} through iterative denoising processes. 
Some works \cite{wu2019pd, liu2023generative, liu2025generative} have utilized the stochastic characteristics of generative models to enhance recommendation diversity.
\\
\textbf{Comparisons of Attention Mechanisms in SR.}
Attention mechanisms of Transformer have played a pivotal role in advancing SR models. Seminal work SASRec \cite{kang2018self} lay the foundation by introducing self-attention to effectively model interaction sequences. Building on this, probabilistic attention mechanisms \cite{liu2024probabilistic},  STOSA \cite{fan2022sequential}, and denoising attention \cite{ho2020denoising} have explored incorporating stochastic properties into attention. Beyond self-attention, additional information has been incorporated to enrich attention modeling, including contextual information \cite{yuan2020attention, huang2018csan}, temporal signals \cite{li2020time,  tran2023attention}, and attributes \cite{rashed2022context, zhang2019feature}.
Structural modifications have also emerged, with hierarchical attention \cite{ying2018sequential, liu2024pay}, memory modules \cite{tan2021dynamic} modification, and frequency-based attention mechanisms \cite{du2023frequency} extending the flexibility of these models. In addition, generative processes have also been combined, such as in VSAN \cite{zhao2021variational}, which models vectors as probability densities via variational inference. Other works have incorporated adversarial learning during Transformer training to generate next-item predictions \cite{ren2020sequential}, while some frameworks use self-attention as the encoder for generative models \cite{du2023idnp, liu2024mmgrec} or as approximators in hybrid architectures \cite{zolghadr2024generative, li2023diffurec}.


\vspace{-1.4mm}
\section{COMPARISONS AND DISCUSSIONS}
\label{sec:preli}
\vspace{-1.1mm}

This section provides a solid foundation for our proposed generative attention by comparisons and theoretical analysis.
Given a set of users $\mathcal{U}$ and a set of items $\mathcal{V}$, along with their interaction histories, SR organizes each user's interaction history into a chronological sequence. For a user $u \in \mathcal{U}$, this sequence is denoted as $\mathcal{S}^u=\left[v_1^u, v_2^u, \ldots, v_{\left|\mathcal{S}^u\right|}^u\right]$, where $v_i^u \in$ $\mathcal{V}$.

\vspace{-1mm}
\subsection{Comparisons of Attention Mechanisms}

For a user's interaction sequence $\mathcal{S}^u$ with a maximum allowed length $n$, sequences exceeding this length are truncated from the start, while shorter ones are padded with zeros to create a uniform sequence $\boldsymbol{s}=\left(s_1, s_2, \ldots, s_n\right)$. Each item in $\mathcal{V}$ is embedded into a latent space through the embedding matrix $\mathbf{E} \in \mathbb{R}^{|\mathcal{V}| \times d}$, where $d$ is the embedding dimension. To enrich the sequence representation with positional information, a trainable positional embedding matrix $\mathbf{P} \in \mathbb{R}^{n \times d}$ is added, resulting in the final sequence encoding:
\begin{equation}
\mathbf{M}_{\mathcal{S}^u}=\left[\mathbf{e}_{s_1}+\mathbf{p}_1, \mathbf{e}_{s_2}+\mathbf{p}_2, \ldots, \mathbf{e}_{s_n}+\mathbf{p}_n\right],
\label{equ:pos-emb}
\end{equation}
where $\mathbf{e}_{s_i}$ represents the embedding of the $i$-th item, and $\mathbf{p}_i$ denotes its corresponding positional embedding. 

\vspace{-1mm}
\subsubsection{\textbf{Traditional Attention}} 
A typical Transformer-style deterministic  mechanism computes attention weights $\mathbf{A}_{\text {det }}$ as:

\begin{equation}
\mathbf{A}_{\mathrm{det}}=\operatorname{softmax}\left(\frac{\mathbf{Q} \mathbf{K}^{\top}}{\sqrt{d}}\right), 
\label{equ:trad-att}
\end{equation}
where $\mathbf{Q} \in \mathbb{R}^{m \times d}, \mathbf{K} \in \mathbb{R}^{m \times d}$ represent the Query and Key matrices, respectively. 
However, this mechanism is deterministic, meaning that the attention matrix $\mathbf{A}_{\text {det }}$ is computed in a fixed manner for each input sequence. This rigidity limits the model's ability to dynamically adjust to evolving user behavior or capture intricate, non-linear dependencies inherent in SR tasks.

\vspace{-1mm}
\subsubsection{\textbf{Generative Attention}}
In generative attention, we treat the attention weights as random variables whose distribution is learned. Formally, we introduce a (latent) variable $\mathbf{z}$ and define:
\begin{equation}
   \mathbf{A}_{\mathrm{gen}} \sim p(\mathbf{A} \mid \mathbf{z}, X) 
   \label{equ:gen-att}
\end{equation}
where $X$ denotes the input sequence, and $\mathbf{z}$ represents latent factors. By sampling from this probability distribution,  generative attention incorporates stochasticity and variability, allowing for the adaptive generation of attention weights based on learned latent factors, which reflect the inherent variability and uncertainty in user behavior. The distribution $p(\cdot)$ is parameterized by a generative model (e.g., VAE or Diffusion Models).

\vspace{-1mm}
\subsection{Theoretical Foundation}

Below is a theoretical derivation proving that generative attention outperforms traditional attention in terms of expressiveness and its ability to handle stochasticity.

\begin{theorem}[Expressiveness]
Let $\mathcal{F}_{\text{det}}$ be the function class corresponding to deterministic attention and $\mathcal{F}_{\text{gen}}$ be the class of generative attention distributions parameterized by a latent variable $\mathbf{z}$. Assume $\mathbf{z}$ is drawn from a continuous latent space $\mathcal{Z} \subseteq \mathbb{R}^d$. Then, under standard smoothness and universal approximation conditions (\textit{e.g.}, neural networks with sufficient width), we have:
$\mathcal{F}_{\text{det}} \subseteq \mathcal{F}_{\text{gen}}$,
indicating that generative attention can represent a strictly larger family of attention distributions than deterministic mechanisms.
\end{theorem}

\begin{proof}[Proof]
In deterministic attention, each input sequence $X$ yields exactly one attention matrix $\mathbf{A}_{\text{det}}(X)$. This corresponds to a single function $f_\theta: X \mapsto \mathbf{A}_{\text{det}}$ in a function space $\mathcal{F}_{\text{det}}$.

In contrast, generative attention introduces a latent variable $\mathbf{z}$. One can view $\mathbf{A}_{\text{gen}}$ as a distribution over $\mathbf{z}$ (\textit{i.e.}, $\mathbf{z} \sim q(\mathbf{z} \mid X)$) and a conditional mapping $g_\phi: (\mathbf{z}, X) \mapsto \mathbf{A}_{\text{gen}}$. The attention weights become samples from a mixture (or family) of possible functions:
\begin{equation}
    \mathbf{A}_{\text{gen}} \sim \int g_\phi(\mathbf{z}, X) q(\mathbf{z} \mid X) d\mathbf{z}.
    \label{equ:pro-expres}
\end{equation}
Since $\mathbf{z}$ is drawn from a continuous space, the family of realizable distributions is strictly larger than a single deterministic map.
Neural networks $g_\phi$ can approximate continuous functions on compact sets arbitrarily well (Universal Approximation Theorem \cite{lu2020universal}). Hence, by allowing a probabilistic mixture over $\mathbf{z}, \mathcal{F}_{\text {gen }}$ can represent a strictly greater variety of attention transformations than the single function in $\mathcal{F}_{\text {det }}$.

Therefore, $\mathcal{F}_{\text {det }} \subseteq \mathcal{F}_{\text {gen }}$. This increased expressiveness is advantageous in SR tasks, where user behavior is dynamic and  discrete, as it enables the model to adapt to shifting patterns and capture the non-linear complexities of user preferences.
Intuitively, generative attention can recover deterministic attention as a special case (by collapsing the latent distribution).
\end{proof}


\begin{theorem}[Stochasticity]
    Let $\mathcal{P}_{\text {det }}$ be the set of all probability distributions induced by deterministic attention (i.e., a delta distribution around $\mathbf{A}_{\text {det }}$ ) and $\mathcal{P}_{\text {gen }}$ be the set of distributions realizable by a generative attention mechanism with latent variable $\mathbf{z}$. Then $\mathcal{P}_{\mathrm{det}} \subset \mathcal{P}_{\mathrm{gen}}$, indicating that generative attention can encode user-driven randomness or noise in attention weights, while deterministic attention essentially disregards such stochasticity.
\end{theorem}

\begin{proof}[Proof]
    Deterministic attention yields $\mathbf{A}_{\text {det }}(X)$ for each input $X$. From a distributional perspective, this can be viewed as
\begin{equation}
    \mathcal{P}_{\operatorname{det}}(\mathbf{A} \mid X)=\delta\left(\mathbf{A}-\mathbf{A}_{\operatorname{det}}(X)\right),
    \label{equ:pro-stoc1}
\end{equation}
where $\delta$ is Dirac delta function \cite{engquist2005discretization}.
In generative attention, we allow $\mathbf{z} \sim q(\mathbf{z} \mid X)$ to be drawn from a non-degenerate distribution:
\begin{equation}
    \mathbf{A}_{\mathrm{gen}} \sim \mathcal{P}_{\operatorname{gen}}(\mathbf{A} \mid X)=\int p(\mathbf{A} \mid \mathbf{z}, X) q(\mathbf{z} \mid X) d \mathbf{z}.
    \label{equ:pro-stoc2}
\end{equation}
This formulation can model a potentially uncountably infinite set of distributions, enabling variability in attention across repeated observations of the same input $X$.

Deterministic attention is recovered as a special case when $q(\mathbf{z} \mid X)$ is a delta distribution, \textit{i.e.}, all mass is concentrated at a single $\mathbf{z}^*$. By permitting variance in $\mathbf{z}$, the generative model can produce diverse, stochastic attention outcomes, thus capturing real-world phenomena (\textit{e.g.}, user uncertainty, inconsistent behaviors).
\end{proof}

\vspace{-2mm}
\section{METHODOLOGY}
\label{sec:Methodology}
\vspace{-.5mm}

This section provides a detailed explanation of how unsupervised latent distribution learning is used to express transformations, enabling the generation of attention distributions. 
As shown in \Cref{fig:illustration-1}, the need for transformations in attention mechanisms is eliminated. To enhance the adaptability of GenAtt, we have specifically designed multi-head and multi-layer configurations.

\begin{figure*}
\centering
    \includegraphics[width=.95\linewidth]{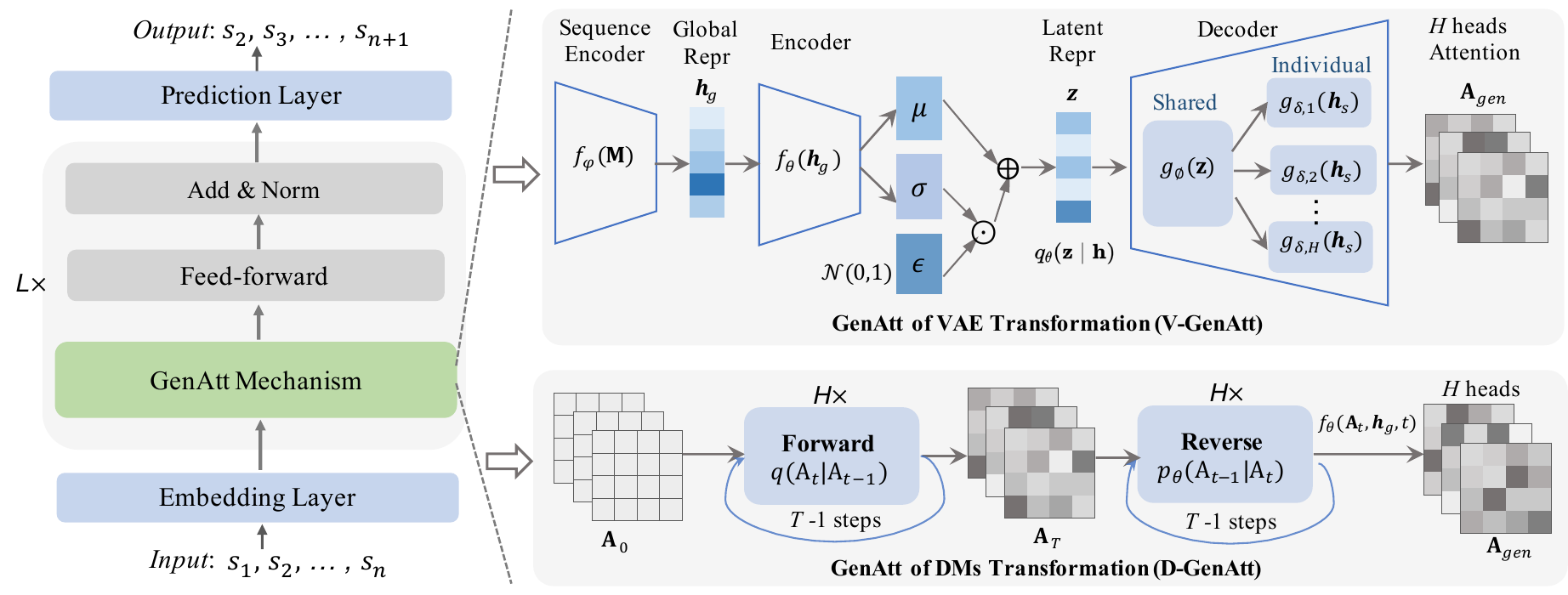}
    \vspace{-2.8mm}
\caption{Visualization of Architecture for Generating Attention Distributions. }
\vspace{-2.6mm}
\label{fig:illustration-1}
\end{figure*}

\subsection{VAE Implementation of GenAtt}

A VAE-based generative model (\textit{i.e.}, V-GenAtt) consists of two main components: the Encoder, which maps the input representations to a distribution over latent variables, and the Decoder, which uses these latent variables to generate new distributions. To make V-GenAtt well-suited for SR tasks and capable of generating stochastic attention distributions that capture complex sequence dependencies, we design novel encoder and decoder components. In particular, we first introduce a sequence encoder that takes into account the temporal dependencies within the sequence and encodes the entire sequence into a latent space,  which is formalized as: 
\begin{equation}
    \mathbf{S}, \mathbf{h}_{g}=f_{\varphi}(\mathbf{\mathbf{M}_{\mathcal{S}^u}}),
\end{equation}
where $\mathbf{S}$ is the sequence-level representation, $\mathbf{M}_{\mathcal{S}^u}$ is the encoding of the input sequence from user $u$, and $\mathbf{h}{g}$ is the global representation that summarizes the entire sequence. The function $f_{\varphi}$ can be implemented using GRU neural networks, which allow the model to effectively capture the sequential nature of the data and produce a rich latent space representation for further processing.

The following encoder is expressed as follows:
\begin{equation}
    \mathbf{\mu}, \log \mathbf{\sigma}^2=f_\theta\left(\mathbf{h}_{g}\right),
\end{equation}
where $f_\theta$ is a neural network parameterized by $\theta$. The latent variable $\mathbf{z}$ is sampled using the reparameterization trick:
$\mathbf{z}=\mu+\sigma \odot \epsilon$, 
where $\epsilon \sim \mathcal{N}(0, \mathbf{I})$ is random noise and $\sigma=\exp \left(\frac{\log \sigma^2}{2}\right)$.   

In the V-GenAtt decoder, the overall process can be described as a combination of a shared encoder followed by individual encoder for each attention head. 
Given the latent variable $\mathbf{z} \in \mathbb{R}^{B \times d_{h}}$ (with $B$ as batch size 
 and $d_{h}$ as the dimensionality of the latent space), the shared encoder produces a shared representation:
   $ \mathbf{h}_{s}=g_{\phi}(\mathbf{z}),$
where $g_{\phi}$ is the shared MLP function.

For each attention head $h$, an individual decoder function $g_{\theta, h}$ projects $\mathbf{h}_{s} \in \mathbb{R}^{B \times d_{h}}$ into the attention matrix $\mathbf{A}_h$:
\begin{equation}
\mathbf{A}_h=g_{\delta, h}\left(\mathbf{h}_{s}\right),
\end{equation}
where $g_{\delta, h}$ is a fully connected layer specific to attention head $h$, and $\mathbf{A}_h \in \mathbb{R}^{B \times n \times n}$ is the attention matrix for head $h$ reshaped from $\mathbb{R}^{B \times n^2}$.
The full multi-head attention distribution $\mathbf{A_{gen}} \in \mathbb{R}^{B \times H \times n \times n}$ is constructed by stacking all the individual attention matrices for each head, and $B$ reflects the batch size.

The global representation $\mathbf{h}_g$, derived from the encoder, encapsulates high-level sequence-wide information. The decoder utilizes this rich representation to generate the attention weights, which are learned through a generative process. Rather than applying a linear and fixed transformation, VAE allows the attention matrix to be dynamically generated based on the latent space learned during training. This means that the attention distribution reflects the model's understanding of how different sequence elements interact, with the learned latent variables guiding the attention mechanism.
VAE inherently approximates the true posterior distribution of the data, and through this approximation, V-GenAtt learns to adjust attention distributions in a way that is context-dependent. Each attention head is able to focus on different parts of the sequence depending on the global context. By learning attention distribution in an unsupervised manner, VAE enables GenAtt to model complex, context driven relationships between sequence elements, which is  important for SR tasks where dependencies can vary significantly.

\subsection{Diffusion Models Implementation of GenAtt}
DMs primarily generate data through a process of gradual noise addition and subsequent denoising \cite{jiang2024diffmm}, which is structured in two stages: the forward process and the reverse process. In the context of GenAtt, this framework is adapted to generate attention distributions. In our DMs implemented generative attention mechanism D-GenAtt, the initial attention matrix $\mathbf{A}_0$ can be set to a zero matrix or a random initialization. This choice does not affect the generative process, as D-GenAtt is designed to learn the distribution of attention weights from the data itself. The key idea is that $\mathbf{A}_0$ serves merely as a starting point, and it is progressively modified through the forward diffusion process. Rather than aiming for the generated attention to match the initial attention, the model learns to generate a flexible, data-dependent attention distribution.

In the forward diffusion process, the clean attention weights $\mathbf{A}_0$ are progressively corrupted by noise over $T$ time steps, and $\beta_t$ controls the variance of the noise at each timestep.
Let $\mathbf{A}_0 \in \mathbb{R}^{B \times \text { H } \times n \times n}$. The forward diffusion process is given by:
\begin{equation}
    \mathbf{A}_t=\sqrt{\alpha_t} \mathbf{A}_0+\sqrt{1-\alpha_t} \epsilon_t \quad \text { for } \quad t=1,2, \ldots, T
\end{equation}
Here, $\alpha_t=1-\beta_t$, $\beta_t \in\left[\beta_{s}, \beta_{e}\right]$ is a noise schedule that determines the amount of noise added at each timestep, with $\beta_{s}$ and $\beta_{e}$ controlling the range of noise, $\epsilon_t \sim \mathcal{N}(0, I)$ is the Gaussian noise added at timestep $t$.
The noisy attention $\mathbf{A}_t$ are computed by applying the noise schedule to  $\mathbf{A}_0$, and the sequence of noisy logits $\mathbf{A}_t$ is generated for $t=1,2, \ldots, T$.

The core of the generative process involves the reverse diffusion process, which is accomplished through a neural network that predicts the noise $\hat{\epsilon}_t$ added at each timestep.
At each timestep $t$, the model predicts the noise $\hat{\epsilon}_t$ given the noisy attention $\mathbf{A}_t$ and the global representation $\mathbf{h}_{g}$ as:
$\hat{\epsilon}_t=f_\theta\left(\mathbf{A}_t, \mathbf{h}_{g}, t\right).$
Using the predicted noise $\hat{\epsilon}_t$, the denoised attention logits $\mathbf{A}_0$ are:
\begin{equation}
    \mathbf{A}_{t-1}=\frac{1}{\sqrt{\alpha_t}}\left(\mathbf{A}_t-\sqrt{1-\alpha_t} \hat{\epsilon}_t\right)
\end{equation}
The denoising process is repeated for $T$ time steps, progressively recovering the clean attention  $\mathbf{A}_0$. The final output after reverse diffusion is the denoised attention matrix $\mathbf{A}_{gen}$, given by:
\begin{equation}
    \mathbf{A}_{gen}=\mathbf{A}_t-\hat{\epsilon}_t \cdot \sqrt{1-\alpha_t}
\end{equation}

During the forward diffusion process, the noisy attention matrix $\mathbf{A}_t$ evolves as Gaussian noise is progressively injected. The global representation $\mathbf{h}_g$ is used to guide the transformation of this noisy matrix into a meaningful attention matrix. It influences the denoising process by providing context, ensuring that dependencies across sequence elements are captured in a manner consistent with the high-level features represented by $\mathbf{h}_g$.
Thus, even though the attention matrix evolves step by step, the global representation ensures that these evolutions are contextually appropriate, based on the sequence as a whole. By conditioning the attention distribution on this global representation, the model can dynamically learn to generate attention matrices that reflect the contextual relationships and dependencies among sequence elements.

\vspace{-1mm}
\subsection{Optimization}
\vspace{-0.5mm}
In SR tasks, the objective function typically compares predicted outcomes with the ground truth, often using binary cross-entropy:
\vspace{-0.5mm}
\begin{equation}
\vspace{-0.5mm}
    \mathcal{L}_{\mathrm{Rec}}=-\sum_{i=1}^N\left[y_i \log \left(\hat{y}_i\right)+\left(1-y_i\right) \log \left(1-\hat{y}_i\right)\right],
    \vspace{-0.6mm}
\end{equation}
where $y_i$ represents the true label, $\hat{y}_i$ is the predicted probability, and $N$ is the number of interactions or items.
Many existing SR models treat Transformer-based architectures as methods for capturing dependencies in sequential data, without considering the potential loss introduced by the Transformer or attention mechanisms. However, our GenAtt perspective necessitates the generative process, which means that typically a generative models related loss is needed to reflect the generative nature. As a result, loss function of GenAtt incorporates two components:
\vspace{-1mm}
\begin{equation}
\vspace{-0.5mm}
    \mathcal{L} =\mathcal{L}_{Rec}+\gamma \mathcal{L}_{Gen}.
    \vspace{-0.5mm}
\end{equation}

As mentioned earlier, the primary goal of the GenAtt generation process is not to make the generated attention distributions match the original ones exactly, but rather to generate attention distributions that can better reflect the complex and dynamic nature of user preferences. Therefore, the objective function for our generative models differs from traditional models by focusing solely on the distribution alignment, rather than exact reconstruction. Specifically, for the VAE-based V-GenAtt, the loss function is: $\mathcal{L}_{Gen}=\operatorname{KL}\left[q_\theta(\mathbf{z} \mid \mathbf{X}) \| p(\mathbf{z})\right]$.
Similarly, for the D-GenAtt, the $\mathcal{L}_{gen}$ loss is: $\mathcal{L}_{Gen}=\mathbb{E}_{\mathbf{A}_t, \epsilon_t}\left[\left\|\epsilon_t-\hat{\epsilon}_t\right\|^2\right]$.


\vspace{-0.6mm}
\section{EXPERIMENTS}
\vspace{-0.5mm}

We comprehensively assess the generative attention perspective in this section by testing two implementations, V-GenAtt and D-GenAtt. The experimental results address the following \textbf{R}esearch \textbf{Q}uestions (\textbf{RQ}s): \textbf{RQ1}: Does GenAtt provide superior recommendation results compared to state-of-the-art baselines?
\textbf{RQ2}: What impact does the generative process in modeling attention distribution have on SR performance?
  \textbf{RQ3}: What are the distinguishing features of GenAtt compared to traditional self-attention mechanisms?
 \textbf{RQ4}: Can GenAtt enhance the expressiveness of SR models?
 \textbf{RQ5}: Does GenAtt have the capability to improve the diversity of recommendations? 
 \textbf{RQ6}: Can the generative process of GenAtt eliminate the need for projection matrices commonly used in attention mechanisms?
     \textbf{RQ7}: Is GenAtt computationally efficient in terms of complexity and parameter requirements?

\vspace{-1mm}
\subsection{Experimental Settings}
\vspace{-0.5mm}

 We evaluate the two implementations of generative attention mechanism using \textbf{four} widely adopted real-world datasets. These datasets cover a variety of categories and exhibit significant variations in matrix densities, reflecting the implicit user-item interactions inherent in recommender systems.
Two subsets, \textbf{Beauty} and \textbf{CDs}, are selected from the extensive Amazon dataset introduced by \cite{he2016ups}. The dataset provides a comprehensive collection of reviews across various product categories. In addition, \textbf{Anime} and \textbf{ML-1M} \cite{harper2015movielens} provide reviews for anime and movies, featuring 43 and 18 categories of items, respectively.  Following common practices in recommendation system evaluation \cite{kang2018self, xie2022contrastive}, users and items with fewer than 10 interactions are excluded to ensure data quality. The widely employed leave-one-out strategy \cite{kang2018self, xie2022contrastive} is adopted to evaluate the performance of each method.

We compare our approaches with \textbf{eleven} state-of-the-art baselines in SR, which are categorized into three groups: 
\vspace{-0.5mm}
\begin{itemize}[leftmargin=*]
\vspace{-0.5mm}
\item \textbf{CNN-based models.} \textbf{Caser} \cite{bai2019personalized} aims to capture high-order patterns by convolutional operations for SR. 
\item \textbf{Transformer-based models.} \textbf{GC-SAN} \cite{xu2019graph} integrates GNN with a self-attention mechanism. \textbf{SASRec} \cite{kang2018self} is a seminal SR method that depends on the attention mechanism. \textbf{BERT4Rec} (Bert) \cite{sun2019bert4rec} employs the bi-directional self-attention mechanism as its backbone. \textbf{CL4SRec} \cite{xie2022contrastive} incorporates CL into SR based on the basic self-attention. \textbf{DuoRec} \cite{qiu2022contrastive} provides a model-level augmentation. \textbf{STRec} \cite{li2023strec} is a recently proposed cross-attention SR model. \textbf{ICSRec} \cite{qin2024intent} generates representations for intentions.
\item \textbf{SR methods that utilize generative models or stochasticity.} \textbf{DiffuRec} \cite{li2023diffurec} is a state-of-the-art method that adapts diffusion models to SR. \textbf{PDRec} \cite{ma2024plug} employs the diffusion
models as a flexible plugin. \textbf{STOSA} \cite{fan2022sequential} is a stochastic attention based method. 
\vspace{-1mm}
\end{itemize}
\vspace{-1mm}

The proposed GenAtt models are implemented using PyTorch, with experiments conducted on an NVIDIA RTX A5000 GPU. To fine-tune the hyperparameters, we perform an extensive grid search across all compared methods, and report performance based on the peak validation results. The embedding dimension is tested for values in the set $\{32, 64, 128\}$, while the maximum sequence length is varied from 10 to 200, with a default length of 50 for our models. The learning rate is optimized within $\{10^{-3}, 10^{-4}\}$. 
For GenAtt models, we explore dropout rates in $\{0.0, 0.1, 0.2, 0.3, 0.4, 0.5, 0.6\}$ and adjust $\gamma$ within $\{0.1, 0.2, 0.4, 0.6, 2.0, 4.0\}$. We set the time steps of D-GenAtt directly equal to the maximum sequence length $n$. The dimensionality of the global representation and the latent representation in the VAE are both set to twice the size of the item embeddings. 
For Transformer-based methods, we investigate the number of layers ($\{1,2,3\}$), the number of attention heads ($\{1,2,4\}$). 
An early stopping strategy halts training if NDCG@20 on the validation set does not improve for 20 epochs.

We employ three primary accuracy metrics: NDCG@$N$, Recall@$N$, and Mean Reciprocal Rank (MRR). To assess GenAtt's impact on diversity, we also incorporate two diversity-focused metrics: Category Coverage (CC@$N$) \cite{puthiya2016coverage, wu2019pd} and Intra-list Distance (ILD@$N$) \cite{puthiya2016coverage, zhang2008avoiding}. These evaluations are conducted for \(N \in\{5,10, 20\}\).

\vspace{-1mm}
\subsection{Results Analysis}

\begin{table*}[tp]
\centering
  \fontsize{7.5}{8.8}\selectfont
  \caption{Overall Performance Comparison. The best results are highlighted in bold, the second-best are marked with an asterisk, and the third-best are underlined. Improvements of both GenAtt models are statistically significant with a t-test $p<0.05$. The \textit{improv.} refers to the percentage increase of the better implementations of GenAtt compared to the best baseline.}
  \vspace{-2.6mm}
  \setlength{\tabcolsep}{1.2mm}{
    \begin{tabular}{ccccccccccccc|cll}
    \toprule[0.9pt]
Dataset&Metric&Caser&SASRec&BERT&STOSA&GC-SAN&CL4SRec&DuoRec&DiffuRec&ICSRec &STRec&PDRec&\textbf{V-GenAtt}& \textbf{D-GenAtt} & \textit{Improv.}
\cr\midrule[0.7pt] 
    \multirow{7}{*}{Beauty} 
    &Recall@5 &0.1096 &0.1293 &0.1148  &0.1302 &0.1204 &0.1312 &0.1329 &0.1390 &\underline{0.1417} &0.1334 & 0.1356 &\textbf{0.1596} &0.1519* &12.63\%  \\
    &Recall@10 &0.1782 &0.1972 &0.1800  &0.2016 &0.1894 &0.1926 &0.1985 &\underline{0.2056} & 0.2029 &0.1997 & 0.2010 & \textbf{0.2350} &0.2237* &14.30\% \\
    &Recall@20 &0.2631 &0.2789 &0.2646  &0.2817 &0.2711 &0.2834 &0.2849 &0.2905 &\underline{0.2920}&0.2800 &0.2837 &\textbf{0.3114} & 0.3084* &6.64\%  \\
    &NDCG@5 &0.0776 &0.0803 &0.0762 &0.0839 &0.0795 &0.0830 &0.0843 &0.0835 & \underline{0.0874}&0.0852 &0.0841 &\textbf{0.1110} & 0.1017*  &27.0\% \\
    &NDCG@10 &0.0823 &0.1056 &0.0890  &0.1133 &0.0967 &0.1104 &0.1125 &0.1170 &\underline{0.1193}
    &0.1152 &0.1097  & \textbf{0.1340} &0.1251* &12,32\% \\
    &NDCG@20 &0.1158 &0.1228 &0.1179 &0.1219 &0.1215 &0.1261 &0.1305 &\underline{0.1354} & 0.1340
    &0.1268 &0.1294 &\textbf{0.1533} &0.1464* &13.22\%  \\ 
    &MRR &0.0803 &0.0879 &0.0826 &0.0910 &0.0898 &0.0907 &0.0945 &0.0971 &\underline{0.0986}
    &0.0924 &0.0915 & \textbf{0.1126} & 0.1088* &14.20\% \\
    \midrule[0.6pt]
    
    \multirow{7}{*}{CDs} 
    &Recall@5 &0.0313 &0.0371 &0.0365 &0.0385 &0.0349 &0.0376 &0.0408 &0.0410 &\underline{0.0423}
    &0.0392 &0.0384 &\textbf{0.0455} &0.0442* &7.57\%  \\
    &Recall@10 &0.0450 &0.0516 &0.0473  &0.0519 &0.0511 &0.0525 &0.0529 &0.0537 & 0.0540 &\underline{0.0543} & 0.0523 & \textbf{0.0584} &0.0570* &7.55\% \\
    &Recall@20 &0.0734 &0.0742 &0.0740  &0.0759 &0.0745 &0.0751 &0.0762 &\underline{0.0771} & 0.0768 &0.0753 &0.0766 & \textbf{0.0838} & 0.0830* &8.69\% \\
    &NDCG@5 &0.0206 &0.0237 &0.0253  &0.0241 &0.0230 &0.0237 &0.0262 &0.0259 & \underline{0.0270}
    &0.0267 &0.0256 &\textbf{0.0298} &0.0274*  &10.37\% \\
    &NDCG@10 &0.0249 &0.0273 &0.0270  &0.0278 &0.0264 &0.0280 &0.0284 &0.0287 & \underline{0.0291}
    &0.0279 & 0.0285 & \textbf{0.0326} &0.0322* &12.03\% \\
    &NDCG@20 &0.0328 &0.0339 &0.0345 &0.0362 &0.0337 &0.0346 &0.0354 &\underline{0.0363} & 0.0357
    &0.0360 &0.0343 & 0.0384* & \textbf{0.0388} &6.89\% \\ 
    &MRR &0.0214 &0.0223 &0.0230  &0.0235 &0.0225 &0.0233 &0.0236 &0.0241 & \underline{0.0249}
    &0.0246 &0.0238 & \textbf{0.0277} &0.0268* &11.24\% \\
    \midrule[0.6pt]
    
    \multirow{7}{*}{Anime} 
    &Recall@5 &0.2672 &0.2902 &0.2847 &0.2860 & 0.2898 &0.2909 &\underline{0.2922} &0.2913 & 0.2919&0.2873 &0.2916 & \textbf{0.3125} & 0.3120* &6.95\%  \\
    &Recall@10 &0.3769 &0.4105 &0.3942  &0.4100 &0.4054 &0.4110 &0.4113 &0.4124 & \underline{0.4133}&0.4120 & 0.4116 & 0.4285* &\textbf{0.4396} & 6.36\% \\
    &Recall@20 &0.5899 &0.6010 &0.5873 &0.5914 &0.5910 &0.5927 &0.6002 &\underline{0.6061} & 0.6055 &0.5912 &0.6025 & 0.6325* & \textbf{0.6329} &4.42\% \\
    &NDCG@5 &0.1758 &0.2035 &0.1895 &0.2040 &0.2006 &0.2023 &0.2051 &0.2064 & \underline{0.2067}
    &0.2060 &0.2054 & \textbf{0.2276} & 0.2184* &10.11\% \\
    &NDCG@10 &0.2092 &0.2350 &0.2234  &0.2317 &0.2290 &0.2361 &0.2357 &0.2373 & \underline{0.2383}
    &0.2365 &0.2370 & \textbf{0.2581} & 0.2562* &8.31\% \\
    &NDCG@20 &0.2804 &0.2914 &0.2817 &0.2901 &0.2905 &0.2921 &\underline{0.2956} &0.2949 & 0.2945
    &0.2916 &0.2937 & 0.3107* & \textbf{0.3130} &5.89\% \\ 
    &MRR &0.1893 &0.2098 &0.1920  &0.2075 &0.2049 &0.2102 &0.2110 &0.2115 & \underline{0.2120}
    &0.2109 &0.2105 & \textbf{0.2241} &0.2223* &5.71\% \\
    \midrule[0.6pt]
    
    \multirow{7}{*}{ML-1M}  
    &Recall@5 &0.0759 &0.0791 &0.0726 &0.0780 &0.0775 &0.0795  &0.0812 &\underline{0.0833} & 0.0828
    &0.0797 &0.0803 & \textbf{0.0933} &0.0903* & 12.00\% \\
    &Recall@10 &0.1377 &0.1429 &0.1383  &0.1440 &0.1424 &0.1435 &0.1426 &0.1439 & \underline{0.1445} &0.1432 & 0.1420 &\textbf{0.1560}& 0.1543*& 7.96\% \\
    &Recall@20  &0.1870 &0.2015 &0.1864 &0.2018 &0.1923 &0.2010 &0.2053 &0.2037 & \underline{0.2060} &0.2029 &0.2021 &\textbf{0.2232} & 0.2207*& 8.35\% \\
    &NDCG@5 &0.0490 &0.0505 &0.0480 &0.0493 &0.0491 &0.0502  &0.0518 &\underline{0.0539} & 0.0536
    &0.0523 &0.0527 & \textbf{0.0599} & 0.0582* & 11.13\%  \\
    &NDCG@10 &0.0684 &0.0710 &0.0686 &0.0719 &0.0706 &0.0720 &0.0734 &0.0742 & \underline{0.0751}
    &0.0731 & 0.0728 & 0.0776 &\textbf{0.0785} & 4.53\% \\
    &NDCG@20 &0.0826 &0.0849 &0.0796 &0.0836 &0.0845 &0.0860  &0.0859 &\underline{0.0875} & 0.0871
    &0.0864 &0.0860 & 0.0936* & \textbf{0.0941} & 7.54\% \\  
    &MRR &0.0586 &0.0614 &0.0594 &0.0610 &0.0608 &0.0612 &0.0628 &0.0632 & \underline{0.0640}
    &0.0625 &0.0618 &\textbf{0.0674}  &0.0670* &5.31\% \\
    \bottomrule[0.9pt]
    \end{tabular}}
    \vspace{-1.5mm}
    \label{tab:all-performance}
\end{table*}

From \Cref{tab:all-performance}, we can draw the following key observations:
\vspace{-1mm}
\begin{itemize}[leftmargin=*]
\vspace{-0.5mm}
\item Both the VAE and DMs implementations of GenAtt show significant advantages across different datasets and evaluation metrics,  demonstrating the superior performance of GenAtt in SR tasks. These results substantiate the claims made in our theoretical analysis and provide strong evidence for GenAtt's effectiveness. This directly answers \textbf{RQ1}, confirming the obvious advantages of generative attention approach. Moreover, the significant improvement over a range of Transformer architectures also partially addresses \textbf{RQ2}, showing that GenAtt provides a better solution for optimizing recommendation outcomes.
\item The GenAtt implementations consistently outperform all baselines across different application domains and varying degrees of sparsity. These findings answer \textbf{RQ4} by demonstrating that GenAtt can better express latent patterns, regardless of the dataset's sparsity. This leads to improved recommendation performance compared to traditional Transformer-based models. This also partially answers \textbf{RQ6}, as our generative attention outperforms existing models relying on the Query-Key-Value architecture, without requiring transformation matrices. 
\item Overall, V-GenAtt outperforms D-GenAtt, primarily because VAE models excel at learning smooth and continuous latent distributions, which are well-suited for generating adaptive attention. On the other hand, D-GenAtt can perform better with longer recommendation lists (larger $n$). This is due to the strength of DMs in progressively refining their attention distributions for better representation of long-term user preferences, which is beneficial when generating longer recommendation lists.
\item The improvements are noticeable on sparse datasets, such as on Beauty. This is largely due to GenAtt's ability to \underline{model stochasticity} through generative attention, which enhances its robustness in capturing relevant dependencies, even with limited interaction data. In contrast, traditional attention mechanisms often struggle with sparse data, leading to suboptimal recommendations. 
\end{itemize}

\Cref{fig:loss-weight} shows the impact of varying the loss weight $\gamma$ on the GenAtt models. The results demonstrate that as $\gamma$ increases, performance steadily improves, indicating that a higher emphasis on the generative process strengthens the model's ability to capture the stochastic nature of user preferences. This allows the model to better account for variability in user behavior, especially in scenarios where preferences are dynamic or uncertain. However, when $\gamma$ becomes too large, the generative loss dominates, causing the model to neglect key signals from the task-specific loss, leading to a drop in performance. This confirms that the generative process, while fulfilling the role of traditional transformations, does so in a more flexible and adaptive manner, capturing user preferences more effectively. This is relevant to \textbf{RQ2}, as it explains the role of the generative process, and further validates \textbf{RQ4} and \textbf{RQ6}, demonstrating that GenAtt achieves strong performance despite not relying on transformation matrices.
Moreover, we directly set loss weight $\gamma=1$ without fine-tuning, showing that GenAtt achieves superior recommendation performance without the need for extensive hyper-parameter optimization. These observations are also validated across other datasets.

\begin{figure}
\centering
     \subfigure[V-GenAtt]{
    \includegraphics[width=0.486\linewidth]{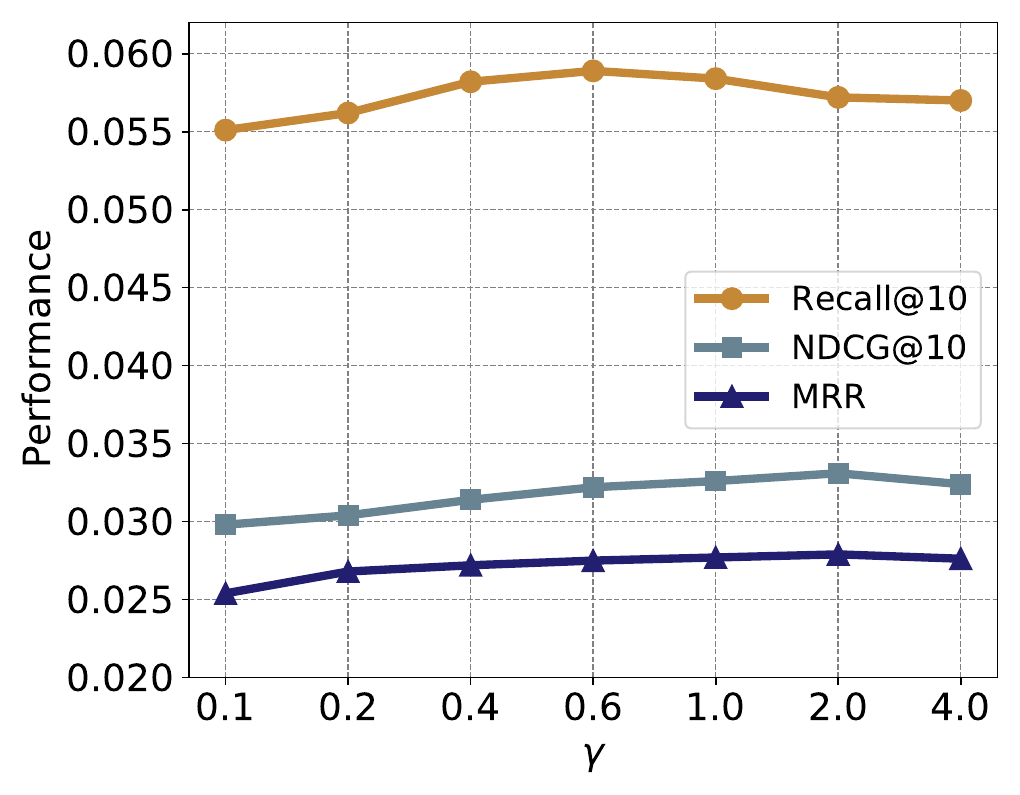}}
    \subfigure[D-GenAtt]{
    \includegraphics[width=0.486\linewidth]{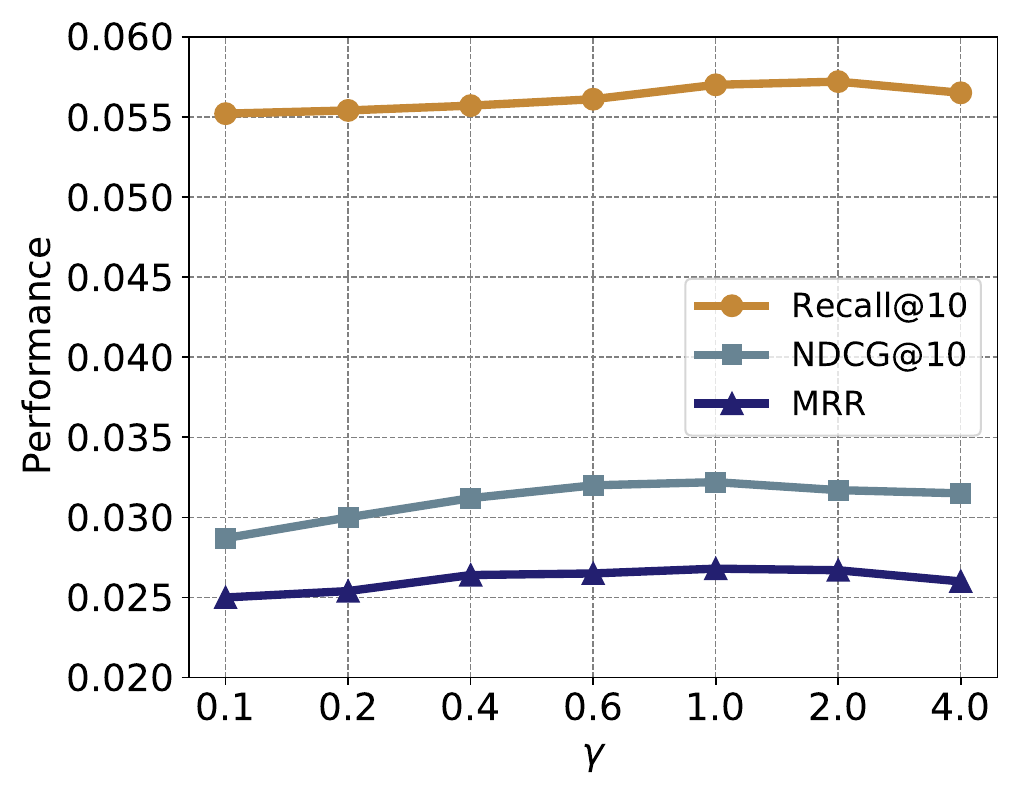}}
\vspace{-3.6mm}
\caption{Performance \textit{w.r.t.} loss weight $\gamma$ on CDs.}
\vspace{-3.5mm}
\label{fig:loss-weight}
\end{figure}

To better address \textbf{RQ6}, we conduct additional experiments by introducing transformations into GenAtt. Specifically, we modify the input to the sequence encoder in the generative model by using a transformed matrix (Query), or by multiplying the generative attention with a transformed Value matrix. The results indicate that while introducing the Query leads to slight improvements in certain cases, the inclusion of the transformed Value matrix actually degrades performance. This could be due to the interaction between the stochastic, non-linear nature of the attention mechanism and the static, transformed Value matrix, which potentially disrupts the learning dynamics.

In \Cref{fig:sequence-length}, we compare the performance of two implementations of GenAtt under varying maximum sequence length $n$. As shown, GenAtt consistently demonstrates a clear advantage across different sequence lengths, further addressing \textbf{RQ1} by highlighting the models' robustness and scalability. Additionally, when the data is sparse, as indicated by smaller values of $n$, GenAtt outperforms the classical model SASRec and the stochastic STOSA, with a more pronounced improvement. This provides further insight into \textbf{RQ3}, showing that a key characteristic of GenAtt is its ability to better handle sparse data and capture the underlying patterns even in the absence of extensive historical information.

\begin{figure}
\centering
     \subfigure[Beauty]{
    \includegraphics[width=0.45\linewidth]{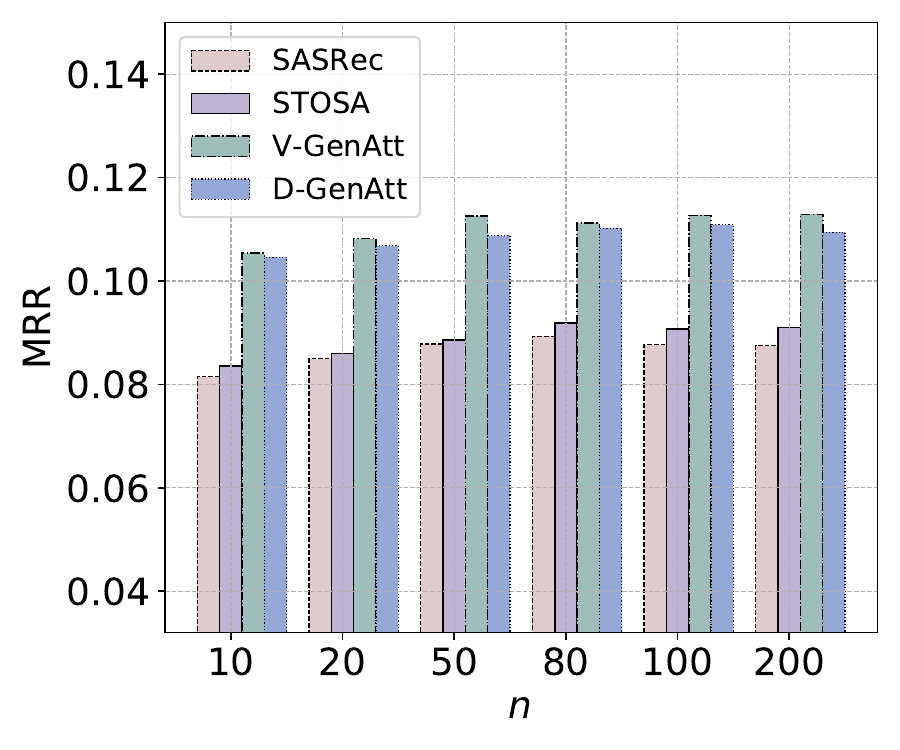}}
    \subfigure[ML-1M]{
    \includegraphics[width=0.45\linewidth]{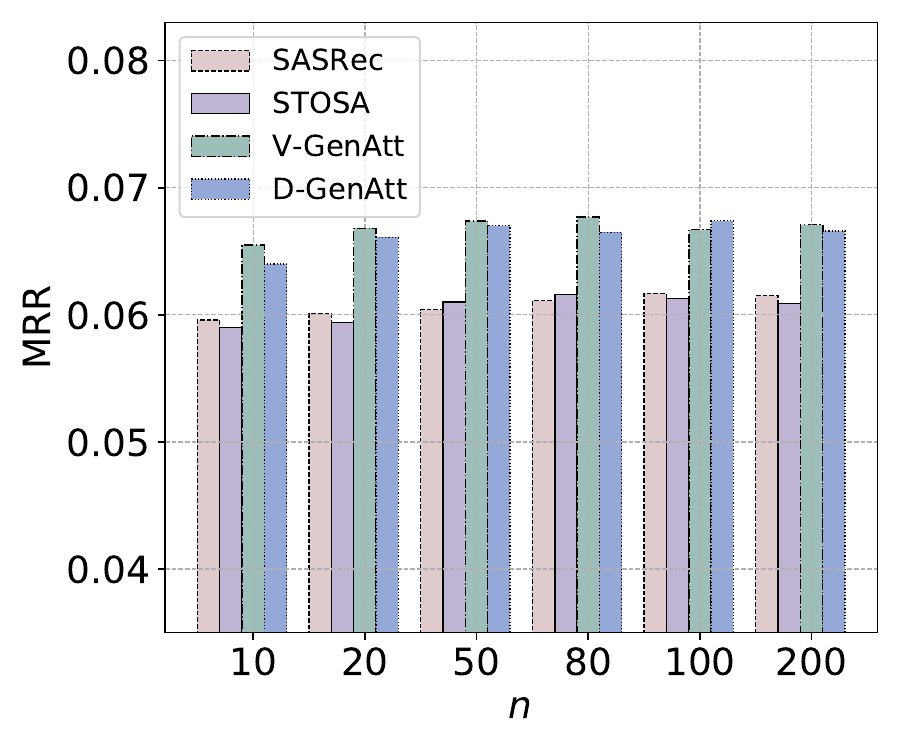}}
\vspace{-4mm}
\caption{Performance \textit{w.r.t.} sequence length $n$.}
\vspace{-4mm}
\label{fig:sequence-length}
\end{figure}

\begin{figure*}
\centering
     \subfigure[V-GenAtt]{
    \includegraphics[width=0.242\linewidth]{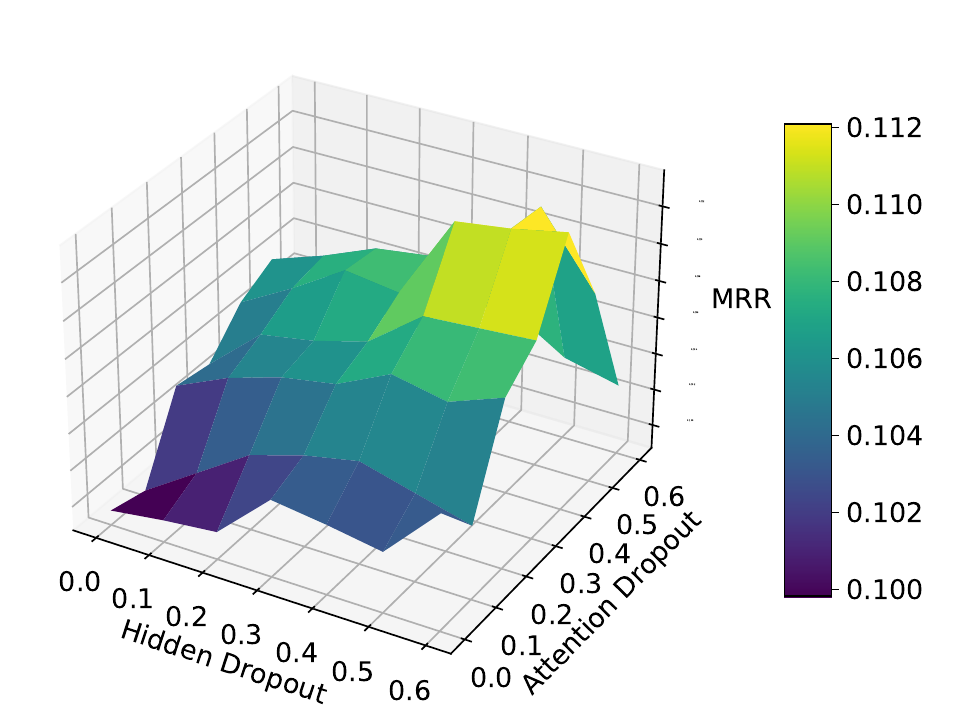}}
    \subfigure[D-GenAtt]{
    \includegraphics[width=0.242\linewidth]{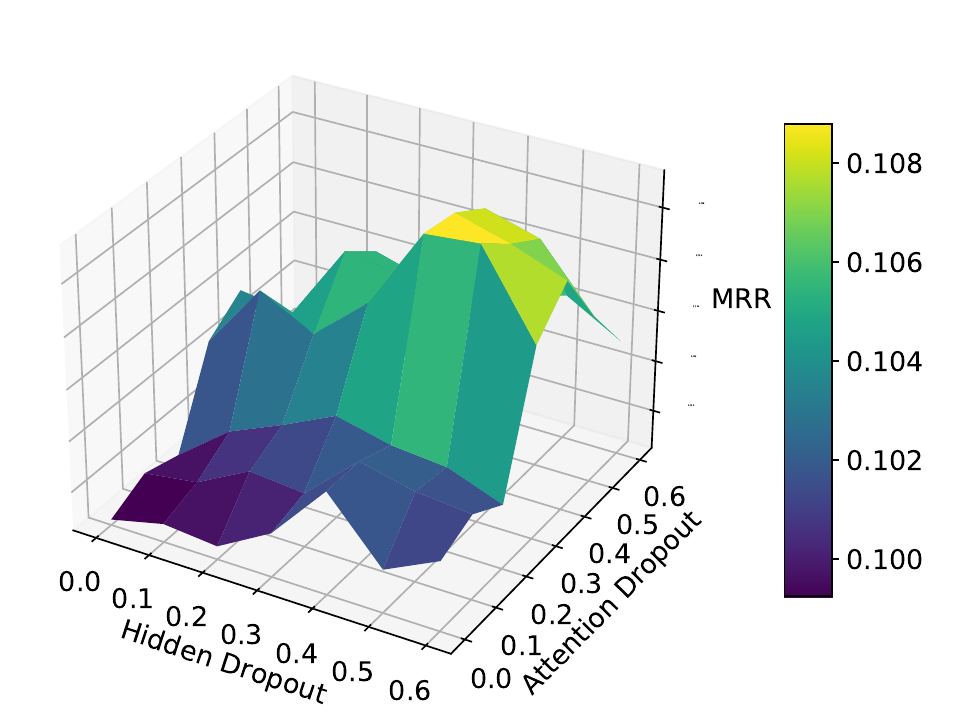}}
    \subfigure[SASRec]{
    \includegraphics[width=0.242\linewidth]{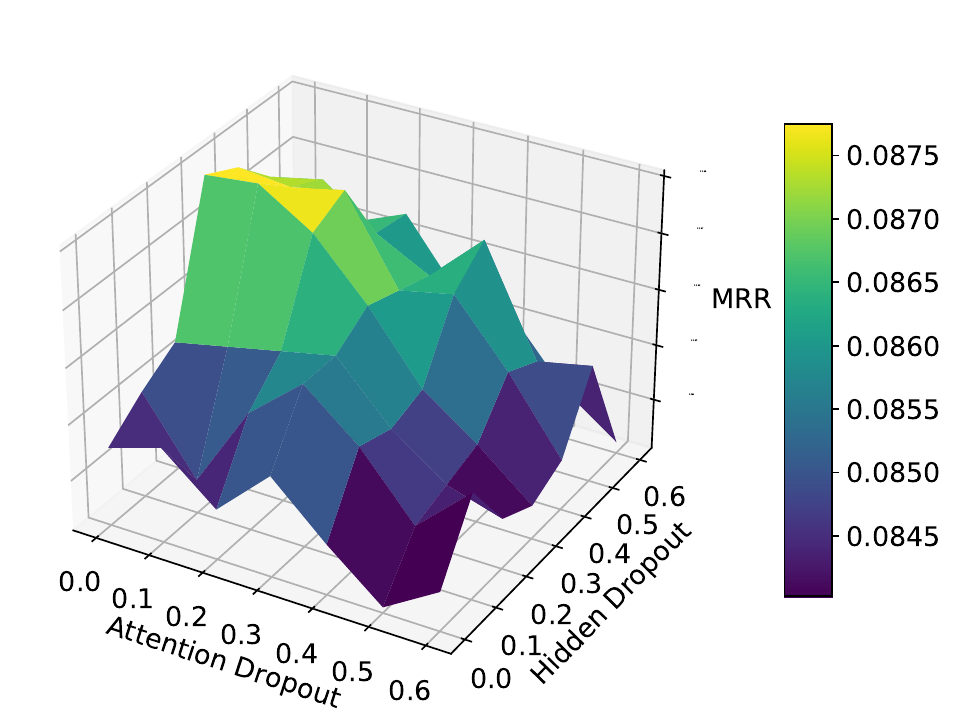}}
    \subfigure[STOSA]{
    \includegraphics[width=0.242\linewidth]{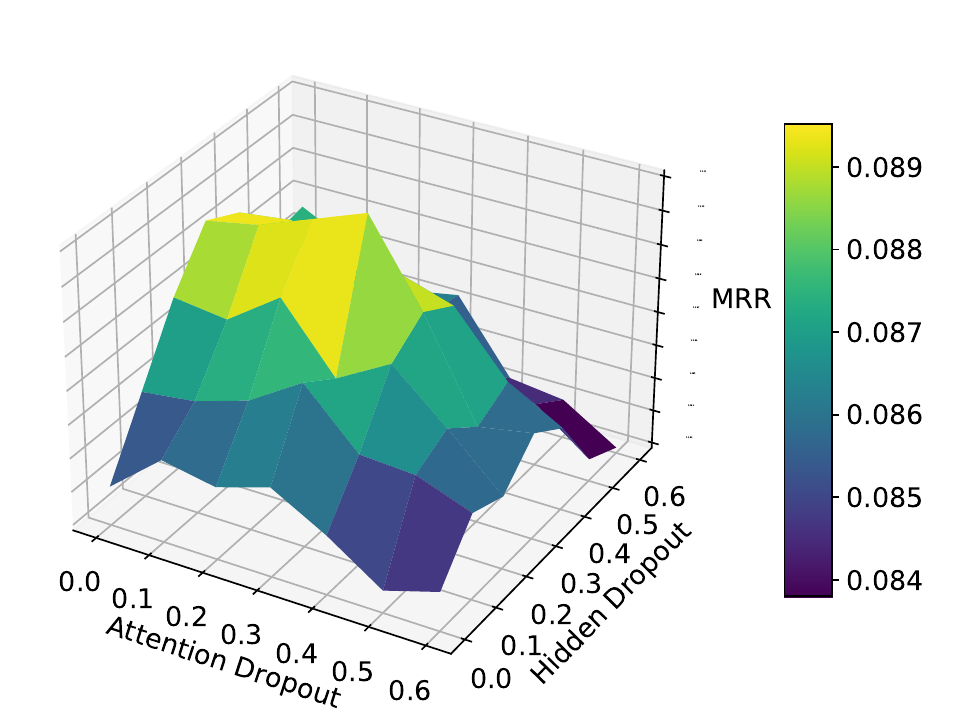}}
\vspace{-3mm}
\caption{Performance \textit{w.r.t.} dropout probabilities on Beauty.}
\vspace{-3.mm}
\label{fig:drop}
\end{figure*}

To further address \textbf{RQ3}, \Cref{fig:drop} compares GenAtt implementations with two representative self-attention-based SR models under different dropout settings. In Transformer-based SR models, dropout is usually applied to hidden representations and attention weights. As shown, GenAtt models perform better with higher dropout values, indicating that GenAtt prefers larger dropout rates. This is because, unlike traditional transformation models, the generative attention mechanism helps regularize the model by introducing more variability in the attention weights, which reduces overfitting and improves generalization. In our experiments, for different datasets, both hidden dropout and attention dropout are directly set to 0.4, which indicates that our model does not require rigorous tuning of hyper-parameter values.

We also validate other properties of generative attention. In SR models using DMs, Transformer is typically used as Approximator for the underlying distribution. To ensure computational efficiency, we employ a simple sequential container in our experiments, and as demonstrated in \Cref{tab:all-performance}, D-GenAtt achieves notable performance under this configuration. We also explore the use of Transformer-based Approximators, which shows some improvements, but for the sake of maintaining computational efficiency and simplicity, we stick with the simpler neural networks.

\begin{figure}
\centering
     \subfigure[Beauty]{
    \includegraphics[width=0.486\linewidth]{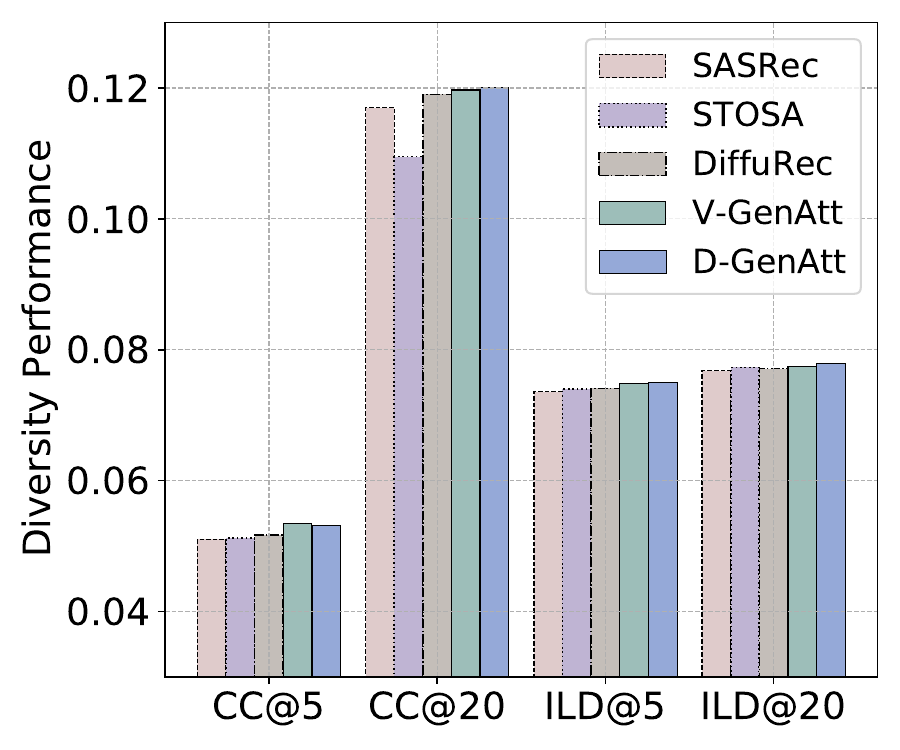}}
    \subfigure[Anime]{
    \includegraphics[width=0.486\linewidth]{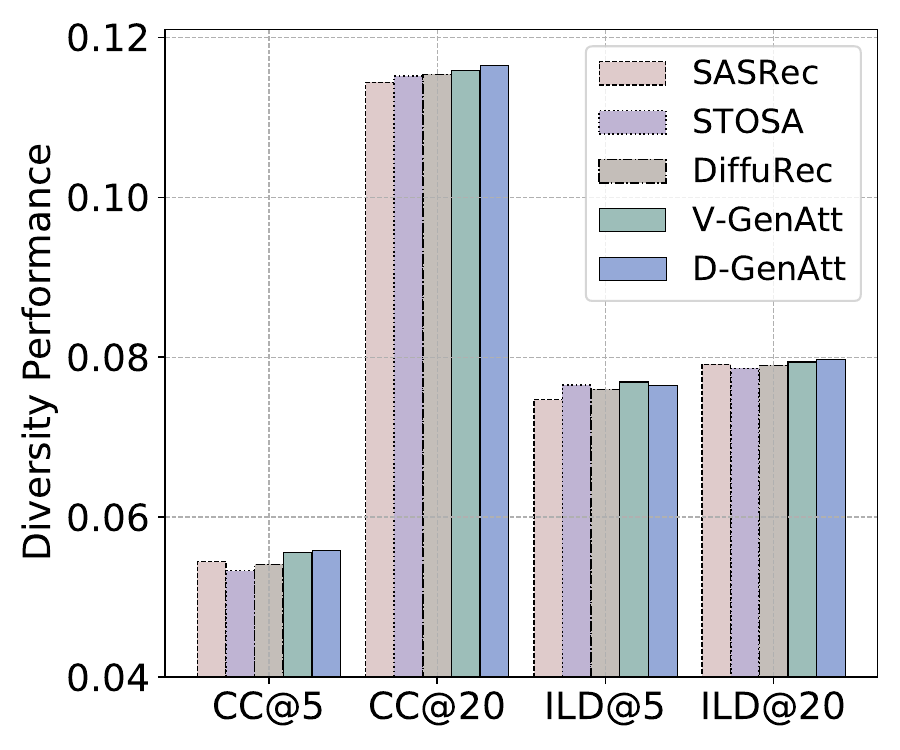}}
\vspace{-3.1mm}
\caption{Diversity performance comparison on two datasets.}
\vspace{-3.7mm}
\label{fig:diveristy}
\end{figure}

\begin{table}
\centering
\setlength{\tabcolsep}{3.4pt} 
\renewcommand{\arraystretch}{1.5} 
\fontsize{7.6}{7.8}\selectfont 
\caption{Average training time (in seconds for one epoch).}
\vspace{-2.7mm}
\begin{tabular}{c|ccccc|ccccc}
\hline
Dataset & \multicolumn{5}{c|}{Beauty} & \multicolumn{5}{c}{ML-1M} \\ \hline
 $n$ & 20 & 30 & 50 & 100 & 200 & 20 & 30 & 50 & 100 & 200 \\
\hline
SASRec &0.43  &0.49  &0.59  &0.74 &0.93  &1.86 &2.01  &2.18  &2.56  &3.52  \\
DiifuRec &0.47  &0.60  &0.73  &0.82  &1.08  &2.35  &2.46  &2.55  &2.89  &4.02  \\
V-GenAtt &0.40 &0.45  &0.60  &0.71  &0.90  &1.78  &1.95  &2.12  &2.58  &3.37  \\
D-GenAtt &0.45 &0.58  &0.75  &0.95  &1.48  &2.31  &2.73  &3.08  &3.46  &4.85 \\
\hline
\end{tabular}
\vspace{-3mm}
\label{tab:training_times}
\end{table}

\Cref{fig:diveristy} is presented to address \textbf{RQ5}, where we compare the diversity performance of three representative models (SASRec, STOSA, and DiffuRec) with our models using diversity metrics (CC and ILD) at their optimal MRR. The findings reveal that stochastic models (STOSA, DiffuRec, and GenAtt) generally outperform the deterministic SASRec in terms of diversity. Among these, GenAtt models outperform the other three models, as they directly use generated attention weights, introducing more stochasticity. Furthermore, D-GenAtt ( incorporates noise through diffusion) achieves superior diversity compared to V-GenAtt (does not explicitly integrate noise). 
These results are related to \textbf{RQ2}, as they demonstrate that GenAtt affects diversity. 

In addition, we explore the relationship between the noise injection level and diversity. In D-GenAtt, the noise intensity is influenced by parameters such as the time steps $T$, $\beta_s, and \beta_e$, which directly control the extent to which noise is introduced. By adjusting these parameters to expand the noise range, we observe an increase in diversity across different datasets, suggesting that greater noise injection facilitates a more varied exploration of the recommendation space. Moreover, we find that the variation in time steps also impacts relevance performance, as it determines how much the model allows user behavior to evolve over time during the diffusion process. Our results show that setting $T$ to 100 or 200 consistently yields better performance across multiple datasets. However, to ensure the generalizability, we choose a time step of 50, equal to $n$. For $\beta_s$ and $\beta_e$, we select the commonly used values of $1 e^{-4}$ and 0.02 , respectively. These results indicate that while tuning the diffusion parameters can enhance model performance, excessive hyperparameter tuning is unnecessary, as the common values are effective across different scenarios.

\Cref{tab:training_times} is used to analyze the training efficiency of GenAtt against representative models. It shows that our GenAtt models achieve higher training efficiency than the competitive DiffuRec model when the sequence length is within a normal range (\textbf{RQ7}). Overall, the training time of generative attention is comparable to existing Transformer-based models, making it a promising approach with notable accuracy and diversity. 

The classical self-attention-based SR model (SASRec) operates with three transformation matrices with the overall space complexity $O\left(|\mathcal{I}| d+n d+3 d^2\right)$. For VAE-based GenAtt, we can summarize the space complexity as: $O\left(|\mathcal{I}| d+n d+n d_h\right)$, where $d_h$ represents the dimensionality of hidden states.
The space complexity for D-GenAtt can be written as: $O\left(|\mathcal{I}| \cdot d+ n^d +T d_h\right)$.
Overall, generative attention models requires less space complexity compared to SASRec, as it avoids the need for multiple transformation matrices. Comparing time complexities reveals distinct computational costs. SASRec incurs $O\left(n^2 d+n d^2\right)$, reflecting both the quadratic complexity in the sequence length ( $n^2$ ) when computing attention scores and the cost of applying transformation matrices $\left(n d^2\right)$. By contrast, VAE-based GenAtt primarily involves encoding the input sequence in $O(n d)$ and generating the attention matrix in $O\left(n^2\right)$, leading to $O\left(n d+n^2\right)$ overall. Meanwhile, Diffusion Models-based generative attention introduces an additional factor $T$ (the number of time steps for the forward and reverse diffusion processes), resulting in $O\left(T \cdot n^2\right)$. As a result, while VAE-based attention can be more efficient than SASRec, diffusion-based attention often becomes slower in practice, particularly for larger $T$.

Combining the actual runtime results in \Cref{tab:training_times} with the complexity analyses provides a comprehensive answer to \textbf{RQ7}.

\vspace{-0.5mm}
\section{CONCLUSION}
\vspace{-0.3mm}

This work introduces a novel perspective on sequential recommendation through the lens of generative attention mechanisms. We have explored two implementations, \textit{i.e.}, V-GenAtt and D-GenAtt, based on VAE and diffusion models, respectively, offering theoretical insights and supporting our findings with comprehensive experimental evaluations. Our results demonstrate the potential of generative models in dynamically learning attention distributions, offering a more expressiveness and flexible alternative to traditional self-attention methods. This approach not only advances the state-of-the-art in recommender systems but also opens up avenues for applications in other domains, such as natural language processing \cite{galassi2020attention} and computer vision \cite{guo2022attention}, where attention mechanisms play a crucial role. Future work will focus on further optimizing these models for scalability, improving their efficiency for large-scale datasets, and exploring their integration with other advanced techniques like reinforcement learning and multi-modal systems.

\vspace{-0.5mm}
\begin{acks}
This work is supported by the Youth Scientific Research Fund of Qinghai University (Grant No.: 2024-QGY-6), and the Quan Cheng Laboratory (Grant No.: QCLZD202301). Weizhi Ma is also sponsored by Beijing Nova Program.
\end{acks}

\bibliographystyle{ACM-Reference-Format}
\bibliography{sample-base}


\begin{thebibliography}{89}


\ifx \showCODEN    \undefined \def \showCODEN     #1{\unskip}     \fi
\ifx \showISBNx    \undefined \def \showISBNx     #1{\unskip}     \fi
\ifx \showISBNxiii \undefined \def \showISBNxiii  #1{\unskip}     \fi
\ifx \showISSN     \undefined \def \showISSN      #1{\unskip}     \fi
\ifx \showLCCN     \undefined \def \showLCCN      #1{\unskip}     \fi
\ifx \shownote     \undefined \def \shownote      #1{#1}          \fi
\ifx \showarticletitle \undefined \def \showarticletitle #1{#1}   \fi
\ifx \showURL      \undefined \def \showURL       {\relax}        \fi
\providecommand\bibfield[2]{#2}
\providecommand\bibinfo[2]{#2}
\providecommand\natexlab[1]{#1}
\providecommand\showeprint[2][]{arXiv:#2}

\bibitem[Bai et~al\mbox{.}(2019)]%
        {bai2019personalized}
\bibfield{author}{\bibinfo{person}{Jinze Bai}, \bibinfo{person}{Chang Zhou}, \bibinfo{person}{Junshuai Song}, \bibinfo{person}{Xiaoru Qu}, \bibinfo{person}{Weiting An}, \bibinfo{person}{Zhao Li}, {and} \bibinfo{person}{Jun Gao}.} \bibinfo{year}{2019}\natexlab{}.
\newblock \showarticletitle{Personalized bundle list recommendation}. In \bibinfo{booktitle}{\emph{The World Wide Web Conference}}. \bibinfo{pages}{60--71}.
\newblock


\bibitem[Boka et~al\mbox{.}(2024)]%
        {boka2024survey}
\bibfield{author}{\bibinfo{person}{Tesfaye~Fenta Boka}, \bibinfo{person}{Zhendong Niu}, {and} \bibinfo{person}{Rama~Bastola Neupane}.} \bibinfo{year}{2024}\natexlab{}.
\newblock \showarticletitle{A survey of sequential recommendation systems: Techniques, evaluation, and future directions}.
\newblock \bibinfo{journal}{\emph{Information Systems}} (\bibinfo{year}{2024}), \bibinfo{pages}{102427}.
\newblock


\bibitem[Bond-Taylor et~al\mbox{.}(2021)]%
        {bond2021deep}
\bibfield{author}{\bibinfo{person}{Sam Bond-Taylor}, \bibinfo{person}{Adam Leach}, \bibinfo{person}{Yang Long}, {and} \bibinfo{person}{Chris~G Willcocks}.} \bibinfo{year}{2021}\natexlab{}.
\newblock \showarticletitle{Deep generative modelling: A comparative review of vaes, gans, normalizing flows, energy-based and autoregressive models}.
\newblock \bibinfo{journal}{\emph{IEEE transactions on pattern analysis and machine intelligence}} \bibinfo{volume}{44}, \bibinfo{number}{11} (\bibinfo{year}{2021}), \bibinfo{pages}{7327--7347}.
\newblock


\bibitem[Cai et~al\mbox{.}(2017)]%
        {cai2017spmc}
\bibfield{author}{\bibinfo{person}{Chenwei Cai}, \bibinfo{person}{Ruining He}, {and} \bibinfo{person}{Julian McAuley}.} \bibinfo{year}{2017}\natexlab{}.
\newblock \showarticletitle{SPMC: socially-aware personalized markov chains for sparse sequential recommendation}. In \bibinfo{booktitle}{\emph{Proceedings of the 26th International Joint Conference on Artificial Intelligence}}. \bibinfo{pages}{1476--1482}.
\newblock


\bibitem[Chen et~al\mbox{.}(2021)]%
        {chen2021modeling}
\bibfield{author}{\bibinfo{person}{Chao Chen}, \bibinfo{person}{Dongsheng Li}, \bibinfo{person}{Junchi Yan}, {and} \bibinfo{person}{Xiaokang Yang}.} \bibinfo{year}{2021}\natexlab{}.
\newblock \showarticletitle{Modeling dynamic user preference via dictionary learning for sequential recommendation}.
\newblock \bibinfo{journal}{\emph{IEEE Transactions on Knowledge and Data Engineering}} \bibinfo{volume}{34}, \bibinfo{number}{11} (\bibinfo{year}{2021}), \bibinfo{pages}{5446--5458}.
\newblock


\bibitem[Chen et~al\mbox{.}(2025)]%
        {chen2025dualcfgl}
\bibfield{author}{\bibinfo{person}{Shuxu Chen}, \bibinfo{person}{Yuanyuan Liu}, \bibinfo{person}{Chao Che}, \bibinfo{person}{Ziqi Wei}, {and} \bibinfo{person}{Zhaoqian Zhong}.} \bibinfo{year}{2025}\natexlab{}.
\newblock \showarticletitle{DualCFGL: dual-channel fusion global and local features for sequential recommendation}.
\newblock \bibinfo{journal}{\emph{Complex \& Intelligent Systems}} \bibinfo{volume}{11}, \bibinfo{number}{1} (\bibinfo{year}{2025}), \bibinfo{pages}{1--18}.
\newblock


\bibitem[Chen et~al\mbox{.}(2022)]%
        {chen2022elecrec}
\bibfield{author}{\bibinfo{person}{Yongjun Chen}, \bibinfo{person}{Jia Li}, {and} \bibinfo{person}{Caiming Xiong}.} \bibinfo{year}{2022}\natexlab{}.
\newblock \showarticletitle{ELECRec: Training sequential recommenders as discriminators}. In \bibinfo{booktitle}{\emph{Proceedings of the 45th International ACM SIGIR Conference on Research and Development in Information Retrieval}}. \bibinfo{pages}{2550--2554}.
\newblock


\bibitem[Cheng et~al\mbox{.}(2025)]%
        {cheng2025sequential}
\bibfield{author}{\bibinfo{person}{Yu Cheng}, \bibinfo{person}{Jiawei Zheng}, \bibinfo{person}{Binquan Wu}, {and} \bibinfo{person}{Qianli Ma}.} \bibinfo{year}{2025}\natexlab{}.
\newblock \showarticletitle{Sequential recommendation via agent-based irrelevancy skipping}.
\newblock \bibinfo{journal}{\emph{Neural Networks}} (\bibinfo{year}{2025}), \bibinfo{pages}{107134}.
\newblock


\bibitem[Correia et~al\mbox{.}(2023)]%
        {correia2023continuous}
\bibfield{author}{\bibinfo{person}{Alvaro~HC Correia}, \bibinfo{person}{Gennaro Gala}, \bibinfo{person}{Erik Quaeghebeur}, \bibinfo{person}{Cassio de Campos}, {and} \bibinfo{person}{Robert Peharz}.} \bibinfo{year}{2023}\natexlab{}.
\newblock \showarticletitle{Continuous mixtures of tractable probabilistic models}. In \bibinfo{booktitle}{\emph{Proceedings of the AAAI Conference on Artificial Intelligence}}, Vol.~\bibinfo{volume}{37}. \bibinfo{pages}{7244--7252}.
\newblock


\bibitem[Cui et~al\mbox{.}(2018)]%
        {cui2018mv}
\bibfield{author}{\bibinfo{person}{Qiang Cui}, \bibinfo{person}{Shu Wu}, \bibinfo{person}{Qiang Liu}, \bibinfo{person}{Wen Zhong}, {and} \bibinfo{person}{Liang Wang}.} \bibinfo{year}{2018}\natexlab{}.
\newblock \showarticletitle{MV-RNN: A multi-view recurrent neural network for sequential recommendation}.
\newblock \bibinfo{journal}{\emph{IEEE Transactions on Knowledge and Data Engineering}} \bibinfo{volume}{32}, \bibinfo{number}{2} (\bibinfo{year}{2018}), \bibinfo{pages}{317--331}.
\newblock


\bibitem[Deldjoo et~al\mbox{.}(2024)]%
        {deldjoo2024review}
\bibfield{author}{\bibinfo{person}{Yashar Deldjoo}, \bibinfo{person}{Zhankui He}, \bibinfo{person}{Julian McAuley}, \bibinfo{person}{Anton Korikov}, \bibinfo{person}{Scott Sanner}, \bibinfo{person}{Arnau Ramisa}, \bibinfo{person}{Ren{\'e} Vidal}, \bibinfo{person}{Maheswaran Sathiamoorthy}, \bibinfo{person}{Atoosa Kasirzadeh}, {and} \bibinfo{person}{Silvia Milano}.} \bibinfo{year}{2024}\natexlab{}.
\newblock \showarticletitle{A review of modern recommender systems using generative models (gen-recsys)}. In \bibinfo{booktitle}{\emph{Proceedings of the 30th ACM SIGKDD Conference on Knowledge Discovery and Data Mining}}. \bibinfo{pages}{6448--6458}.
\newblock


\bibitem[Du et~al\mbox{.}(2023b)]%
        {du2023sequential}
\bibfield{author}{\bibinfo{person}{Hanwen Du}, \bibinfo{person}{Huanhuan Yuan}, \bibinfo{person}{Zhen Huang}, \bibinfo{person}{Pengpeng Zhao}, {and} \bibinfo{person}{Xiaofang Zhou}.} \bibinfo{year}{2023}\natexlab{b}.
\newblock \showarticletitle{Sequential recommendation with diffusion models}.
\newblock \bibinfo{journal}{\emph{arXiv preprint arXiv:2304.04541}} (\bibinfo{year}{2023}).
\newblock


\bibitem[Du et~al\mbox{.}(2023a)]%
        {du2023idnp}
\bibfield{author}{\bibinfo{person}{Jing Du}, \bibinfo{person}{Zesheng Ye}, \bibinfo{person}{Bin Guo}, \bibinfo{person}{Zhiwen Yu}, {and} \bibinfo{person}{Lina Yao}.} \bibinfo{year}{2023}\natexlab{a}.
\newblock \showarticletitle{Idnp: Interest dynamics modeling using generative neural processes for sequential recommendation}. In \bibinfo{booktitle}{\emph{Proceedings of the Sixteenth ACM International Conference on Web Search and Data Mining}}. \bibinfo{pages}{481--489}.
\newblock


\bibitem[Du et~al\mbox{.}(2023c)]%
        {du2023frequency}
\bibfield{author}{\bibinfo{person}{Xinyu Du}, \bibinfo{person}{Huanhuan Yuan}, \bibinfo{person}{Pengpeng Zhao}, \bibinfo{person}{Jianfeng Qu}, \bibinfo{person}{Fuzhen Zhuang}, \bibinfo{person}{Guanfeng Liu}, \bibinfo{person}{Yanchi Liu}, {and} \bibinfo{person}{Victor~S Sheng}.} \bibinfo{year}{2023}\natexlab{c}.
\newblock \showarticletitle{Frequency enhanced hybrid attention network for sequential recommendation}. In \bibinfo{booktitle}{\emph{Proceedings of the 46th International ACM SIGIR Conference on Research and Development in Information Retrieval}}. \bibinfo{pages}{78--88}.
\newblock


\bibitem[Engquist et~al\mbox{.}(2005)]%
        {engquist2005discretization}
\bibfield{author}{\bibinfo{person}{Bj{\"o}rn Engquist}, \bibinfo{person}{Anna-Karin Tornberg}, {and} \bibinfo{person}{Richard Tsai}.} \bibinfo{year}{2005}\natexlab{}.
\newblock \showarticletitle{Discretization of Dirac delta functions in level set methods}.
\newblock \bibinfo{journal}{\emph{J. Comput. Phys.}} \bibinfo{volume}{207}, \bibinfo{number}{1} (\bibinfo{year}{2005}), \bibinfo{pages}{28--51}.
\newblock


\bibitem[Fan et~al\mbox{.}(2022)]%
        {fan2022sequential}
\bibfield{author}{\bibinfo{person}{Ziwei Fan}, \bibinfo{person}{Zhiwei Liu}, \bibinfo{person}{Yu Wang}, \bibinfo{person}{Alice Wang}, \bibinfo{person}{Zahra Nazari}, \bibinfo{person}{Lei Zheng}, \bibinfo{person}{Hao Peng}, {and} \bibinfo{person}{Philip~S Yu}.} \bibinfo{year}{2022}\natexlab{}.
\newblock \showarticletitle{Sequential recommendation via stochastic self-attention}. In \bibinfo{booktitle}{\emph{Proceedings of the ACM web conference 2022}}. \bibinfo{pages}{2036--2047}.
\newblock


\bibitem[Galassi et~al\mbox{.}(2020)]%
        {galassi2020attention}
\bibfield{author}{\bibinfo{person}{Andrea Galassi}, \bibinfo{person}{Marco Lippi}, {and} \bibinfo{person}{Paolo Torroni}.} \bibinfo{year}{2020}\natexlab{}.
\newblock \showarticletitle{Attention in natural language processing}.
\newblock \bibinfo{journal}{\emph{IEEE transactions on neural networks and learning systems}} \bibinfo{volume}{32}, \bibinfo{number}{10} (\bibinfo{year}{2020}), \bibinfo{pages}{4291--4308}.
\newblock


\bibitem[Ge et~al\mbox{.}(2025)]%
        {ge2025personalized}
\bibfield{author}{\bibinfo{person}{Meiling Ge}, \bibinfo{person}{Chengduan Wang}, \bibinfo{person}{Xueyang Qin}, \bibinfo{person}{Jiangyan Dai}, \bibinfo{person}{Lei Huang}, \bibinfo{person}{Qibing Qin}, {and} \bibinfo{person}{Wenfeng Zhang}.} \bibinfo{year}{2025}\natexlab{}.
\newblock \showarticletitle{Personalized Dual Transformer Network for sequential recommendation}.
\newblock \bibinfo{journal}{\emph{Neurocomputing}} (\bibinfo{year}{2025}), \bibinfo{pages}{129244}.
\newblock


\bibitem[Graves and Graves(2012)]%
        {graves2012long}
\bibfield{author}{\bibinfo{person}{Alex Graves} {and} \bibinfo{person}{Alex Graves}.} \bibinfo{year}{2012}\natexlab{}.
\newblock \showarticletitle{Long short-term memory}.
\newblock \bibinfo{journal}{\emph{Supervised sequence labelling with recurrent neural networks}} (\bibinfo{year}{2012}), \bibinfo{pages}{37--45}.
\newblock


\bibitem[Guo et~al\mbox{.}(2022)]%
        {guo2022attention}
\bibfield{author}{\bibinfo{person}{Meng-Hao Guo}, \bibinfo{person}{Tian-Xing Xu}, \bibinfo{person}{Jiang-Jiang Liu}, \bibinfo{person}{Zheng-Ning Liu}, \bibinfo{person}{Peng-Tao Jiang}, \bibinfo{person}{Tai-Jiang Mu}, \bibinfo{person}{Song-Hai Zhang}, \bibinfo{person}{Ralph~R Martin}, \bibinfo{person}{Ming-Ming Cheng}, {and} \bibinfo{person}{Shi-Min Hu}.} \bibinfo{year}{2022}\natexlab{}.
\newblock \showarticletitle{Attention mechanisms in computer vision: A survey}.
\newblock \bibinfo{journal}{\emph{Computational visual media}} \bibinfo{volume}{8}, \bibinfo{number}{3} (\bibinfo{year}{2022}), \bibinfo{pages}{331--368}.
\newblock


\bibitem[Han et~al\mbox{.}(2024)]%
        {han2024efficient}
\bibfield{author}{\bibinfo{person}{Yongqiang Han}, \bibinfo{person}{Hao Wang}, \bibinfo{person}{Kefan Wang}, \bibinfo{person}{Likang Wu}, \bibinfo{person}{Zhi Li}, \bibinfo{person}{Wei Guo}, \bibinfo{person}{Yong Liu}, \bibinfo{person}{Defu Lian}, {and} \bibinfo{person}{Enhong Chen}.} \bibinfo{year}{2024}\natexlab{}.
\newblock \showarticletitle{Efficient Noise-Decoupling for Multi-Behavior Sequential Recommendation}. In \bibinfo{booktitle}{\emph{Proceedings of the ACM on Web Conference 2024}}. \bibinfo{pages}{3297--3306}.
\newblock


\bibitem[Harper and Konstan(2015)]%
        {harper2015movielens}
\bibfield{author}{\bibinfo{person}{F~Maxwell Harper} {and} \bibinfo{person}{Joseph~A Konstan}.} \bibinfo{year}{2015}\natexlab{}.
\newblock \showarticletitle{The movielens datasets: History and context}.
\newblock \bibinfo{journal}{\emph{Acm transactions on interactive intelligent systems (tiis)}} \bibinfo{volume}{5}, \bibinfo{number}{4} (\bibinfo{year}{2015}), \bibinfo{pages}{1--19}.
\newblock


\bibitem[He and McAuley(2016a)]%
        {he2016fusing}
\bibfield{author}{\bibinfo{person}{Ruining He} {and} \bibinfo{person}{Julian McAuley}.} \bibinfo{year}{2016}\natexlab{a}.
\newblock \showarticletitle{Fusing similarity models with markov chains for sparse sequential recommendation}. In \bibinfo{booktitle}{\emph{2016 IEEE 16th international conference on data mining (ICDM)}}. IEEE, \bibinfo{pages}{191--200}.
\newblock


\bibitem[He and McAuley(2016b)]%
        {he2016ups}
\bibfield{author}{\bibinfo{person}{Ruining He} {and} \bibinfo{person}{Julian McAuley}.} \bibinfo{year}{2016}\natexlab{b}.
\newblock \showarticletitle{Ups and downs: Modeling the visual evolution of fashion trends with one-class collaborative filtering}. In \bibinfo{booktitle}{\emph{proceedings of the 25th international conference on world wide web}}. \bibinfo{pages}{507--517}.
\newblock


\bibitem[Ho et~al\mbox{.}(2020)]%
        {ho2020denoising}
\bibfield{author}{\bibinfo{person}{Jonathan Ho}, \bibinfo{person}{Ajay Jain}, {and} \bibinfo{person}{Pieter Abbeel}.} \bibinfo{year}{2020}\natexlab{}.
\newblock \showarticletitle{Denoising diffusion probabilistic models}.
\newblock \bibinfo{journal}{\emph{Advances in neural information processing systems}}  \bibinfo{volume}{33} (\bibinfo{year}{2020}), \bibinfo{pages}{6840--6851}.
\newblock


\bibitem[Huang et~al\mbox{.}(2018)]%
        {huang2018csan}
\bibfield{author}{\bibinfo{person}{Xiaowen Huang}, \bibinfo{person}{Shengsheng Qian}, \bibinfo{person}{Quan Fang}, \bibinfo{person}{Jitao Sang}, {and} \bibinfo{person}{Changsheng Xu}.} \bibinfo{year}{2018}\natexlab{}.
\newblock \showarticletitle{Csan: Contextual self-attention network for user sequential recommendation}. In \bibinfo{booktitle}{\emph{Proceedings of the 26th ACM international conference on Multimedia}}. \bibinfo{pages}{447--455}.
\newblock


\bibitem[Jiang et~al\mbox{.}(2024a)]%
        {jiang2024diffmm}
\bibfield{author}{\bibinfo{person}{Yangqin Jiang}, \bibinfo{person}{Lianghao Xia}, \bibinfo{person}{Wei Wei}, \bibinfo{person}{Da Luo}, \bibinfo{person}{Kangyi Lin}, {and} \bibinfo{person}{Chao Huang}.} \bibinfo{year}{2024}\natexlab{a}.
\newblock \showarticletitle{Diffmm: Multi-modal diffusion model for recommendation}. In \bibinfo{booktitle}{\emph{Proceedings of the 32nd ACM International Conference on Multimedia}}. \bibinfo{pages}{7591--7599}.
\newblock


\bibitem[Jiang et~al\mbox{.}(2024b)]%
        {jiang2024diffkg}
\bibfield{author}{\bibinfo{person}{Yangqin Jiang}, \bibinfo{person}{Yuhao Yang}, \bibinfo{person}{Lianghao Xia}, {and} \bibinfo{person}{Chao Huang}.} \bibinfo{year}{2024}\natexlab{b}.
\newblock \showarticletitle{Diffkg: Knowledge graph diffusion model for recommendation}. In \bibinfo{booktitle}{\emph{Proceedings of the 17th ACM International Conference on Web Search and Data Mining}}. \bibinfo{pages}{313--321}.
\newblock


\bibitem[Kang and McAuley(2018)]%
        {kang2018self}
\bibfield{author}{\bibinfo{person}{Wang-Cheng Kang} {and} \bibinfo{person}{Julian McAuley}.} \bibinfo{year}{2018}\natexlab{}.
\newblock \showarticletitle{Self-attentive sequential recommendation}. In \bibinfo{booktitle}{\emph{2018 IEEE international conference on data mining (ICDM)}}. IEEE, \bibinfo{pages}{197--206}.
\newblock


\bibitem[Khan et~al\mbox{.}(2025)]%
        {khan2025convseq}
\bibfield{author}{\bibinfo{person}{Zeeshan Khan}, \bibinfo{person}{Zafran Khan}, {and} \bibinfo{person}{Naima Iltaf}.} \bibinfo{year}{2025}\natexlab{}.
\newblock \showarticletitle{ConvSeq-MF: Convo-Sequential Matrix Factorization for recommender system}.
\newblock \bibinfo{journal}{\emph{Neurocomputing}}  \bibinfo{volume}{618} (\bibinfo{year}{2025}), \bibinfo{pages}{128932}.
\newblock


\bibitem[Kingma(2013)]%
        {kingma2013auto}
\bibfield{author}{\bibinfo{person}{Diederik~P Kingma}.} \bibinfo{year}{2013}\natexlab{}.
\newblock \showarticletitle{Auto-encoding variational bayes}.
\newblock \bibinfo{journal}{\emph{arXiv preprint arXiv:1312.6114}} (\bibinfo{year}{2013}).
\newblock


\bibitem[Li et~al\mbox{.}(2023b)]%
        {li2023strec}
\bibfield{author}{\bibinfo{person}{Chengxi Li}, \bibinfo{person}{Yejing Wang}, \bibinfo{person}{Qidong Liu}, \bibinfo{person}{Xiangyu Zhao}, \bibinfo{person}{Wanyu Wang}, \bibinfo{person}{Yiqi Wang}, \bibinfo{person}{Lixin Zou}, \bibinfo{person}{Wenqi Fan}, {and} \bibinfo{person}{Qing Li}.} \bibinfo{year}{2023}\natexlab{b}.
\newblock \showarticletitle{STRec: Sparse Transformer for Sequential Recommendations}. In \bibinfo{booktitle}{\emph{Proceedings of the 17th ACM Conference on Recommender Systems}}. \bibinfo{pages}{101--111}.
\newblock


\bibitem[Li et~al\mbox{.}(2025)]%
        {li2025dgt}
\bibfield{author}{\bibinfo{person}{Chenglin Li}, \bibinfo{person}{Tao Xie}, \bibinfo{person}{Chenyun Yu}, \bibinfo{person}{Bo Hu}, \bibinfo{person}{Zang Li}, \bibinfo{person}{Lei Cheng}, \bibinfo{person}{Beibei Kong}, {and} \bibinfo{person}{Di Niu}.} \bibinfo{year}{2025}\natexlab{}.
\newblock \showarticletitle{DGT: Unbiased sequential recommendation via Disentangled Graph Transformer}.
\newblock \bibinfo{journal}{\emph{Knowledge-Based Systems}} (\bibinfo{year}{2025}), \bibinfo{pages}{112946}.
\newblock


\bibitem[Li et~al\mbox{.}(2020)]%
        {li2020time}
\bibfield{author}{\bibinfo{person}{Jiacheng Li}, \bibinfo{person}{Yujie Wang}, {and} \bibinfo{person}{Julian McAuley}.} \bibinfo{year}{2020}\natexlab{}.
\newblock \showarticletitle{Time interval aware self-attention for sequential recommendation}. In \bibinfo{booktitle}{\emph{Proceedings of the 13th international conference on web search and data mining}}. \bibinfo{pages}{322--330}.
\newblock


\bibitem[Li et~al\mbox{.}(2023a)]%
        {li2023diffurec}
\bibfield{author}{\bibinfo{person}{Zihao Li}, \bibinfo{person}{Aixin Sun}, {and} \bibinfo{person}{Chenliang Li}.} \bibinfo{year}{2023}\natexlab{a}.
\newblock \showarticletitle{Diffurec: A diffusion model for sequential recommendation}.
\newblock \bibinfo{journal}{\emph{ACM Transactions on Information Systems}} \bibinfo{volume}{42}, \bibinfo{number}{3} (\bibinfo{year}{2023}), \bibinfo{pages}{1--28}.
\newblock


\bibitem[Liang et~al\mbox{.}(2018)]%
        {liang2018variational}
\bibfield{author}{\bibinfo{person}{Dawen Liang}, \bibinfo{person}{Rahul~G Krishnan}, \bibinfo{person}{Matthew~D Hoffman}, {and} \bibinfo{person}{Tony Jebara}.} \bibinfo{year}{2018}\natexlab{}.
\newblock \showarticletitle{Variational autoencoders for collaborative filtering}. In \bibinfo{booktitle}{\emph{Proceedings of the 2018 world wide web conference}}. \bibinfo{pages}{689--698}.
\newblock


\bibitem[Lin et~al\mbox{.}(2022)]%
        {lin2022structured}
\bibfield{author}{\bibinfo{person}{Zhouhan Lin}, \bibinfo{person}{Minwei Feng}, \bibinfo{person}{Cicero~Nogueira dos Santos}, \bibinfo{person}{Mo Yu}, \bibinfo{person}{Bing Xiang}, \bibinfo{person}{Bowen Zhou}, {and} \bibinfo{person}{Yoshua Bengio}.} \bibinfo{year}{2022}\natexlab{}.
\newblock \showarticletitle{A STRUCTURED SELF-ATTENTIVE SENTENCE EMBEDDING}. In \bibinfo{booktitle}{\emph{International Conference on Learning Representations}}.
\newblock


\bibitem[Liu et~al\mbox{.}(2024c)]%
        {liu2024mmgrec}
\bibfield{author}{\bibinfo{person}{Han Liu}, \bibinfo{person}{Yinwei Wei}, \bibinfo{person}{Xuemeng Song}, \bibinfo{person}{Weili Guan}, \bibinfo{person}{Yuan-Fang Li}, {and} \bibinfo{person}{Liqiang Nie}.} \bibinfo{year}{2024}\natexlab{c}.
\newblock \showarticletitle{MMGRec: Multimodal Generative Recommendation with Transformer Model}.
\newblock \bibinfo{journal}{\emph{arXiv preprint arXiv:2404.16555}} (\bibinfo{year}{2024}).
\newblock


\bibitem[Liu et~al\mbox{.}(2016)]%
        {liu2016context}
\bibfield{author}{\bibinfo{person}{Qiang Liu}, \bibinfo{person}{Shu Wu}, \bibinfo{person}{Diyi Wang}, \bibinfo{person}{Zhaokang Li}, {and} \bibinfo{person}{Liang Wang}.} \bibinfo{year}{2016}\natexlab{}.
\newblock \showarticletitle{Context-aware sequential recommendation}. In \bibinfo{booktitle}{\emph{2016 IEEE 16th International Conference on Data Mining (ICDM)}}. IEEE, \bibinfo{pages}{1053--1058}.
\newblock


\bibitem[Liu et~al\mbox{.}(2023)]%
        {liu2023generative}
\bibfield{author}{\bibinfo{person}{Shuchang Liu}, \bibinfo{person}{Qingpeng Cai}, \bibinfo{person}{Zhankui He}, \bibinfo{person}{Bowen Sun}, \bibinfo{person}{Julian McAuley}, \bibinfo{person}{Dong Zheng}, \bibinfo{person}{Peng Jiang}, {and} \bibinfo{person}{Kun Gai}.} \bibinfo{year}{2023}\natexlab{}.
\newblock \showarticletitle{Generative flow network for listwise recommendation}. In \bibinfo{booktitle}{\emph{Proceedings of the 29th ACM SIGKDD Conference on Knowledge Discovery and Data Mining}}. \bibinfo{pages}{1524--1534}.
\newblock


\bibitem[Liu(2025)]%
        {liu2025generative}
\bibfield{author}{\bibinfo{person}{Yuli Liu}.} \bibinfo{year}{2025}\natexlab{}.
\newblock \showarticletitle{A generative and discriminative model for diversity-promoting recommendation}.
\newblock \bibinfo{journal}{\emph{Information Systems}}  \bibinfo{volume}{128} (\bibinfo{year}{2025}), \bibinfo{pages}{102488}.
\newblock


\bibitem[Liu et~al\mbox{.}(2024a)]%
        {liu2024pay}
\bibfield{author}{\bibinfo{person}{Yuli Liu}, \bibinfo{person}{Min Liu}, {and} \bibinfo{person}{Xiaojing Liu}.} \bibinfo{year}{2024}\natexlab{a}.
\newblock \showarticletitle{Pay Attention to Attention for Sequential Recommendation}. In \bibinfo{booktitle}{\emph{Proceedings of the 18th ACM Conference on Recommender Systems}}. \bibinfo{pages}{890--895}.
\newblock


\bibitem[Liu et~al\mbox{.}(2024b)]%
        {liu2024probabilistic}
\bibfield{author}{\bibinfo{person}{Yuli Liu}, \bibinfo{person}{Christian Walder}, \bibinfo{person}{Lexing Xie}, {and} \bibinfo{person}{Yiqun Liu}.} \bibinfo{year}{2024}\natexlab{b}.
\newblock \showarticletitle{Probabilistic Attention for Sequential Recommendation}. In \bibinfo{booktitle}{\emph{Proceedings of the 30th ACM SIGKDD Conference on Knowledge Discovery and Data Mining}}. \bibinfo{pages}{1956--1967}.
\newblock


\bibitem[Liu et~al\mbox{.}(2024d)]%
        {liu2024selfgnn}
\bibfield{author}{\bibinfo{person}{Yuxi Liu}, \bibinfo{person}{Lianghao Xia}, {and} \bibinfo{person}{Chao Huang}.} \bibinfo{year}{2024}\natexlab{d}.
\newblock \showarticletitle{Selfgnn: Self-supervised graph neural networks for sequential recommendation}. In \bibinfo{booktitle}{\emph{Proceedings of the 47th International ACM SIGIR Conference on Research and Development in Information Retrieval}}. \bibinfo{pages}{1609--1618}.
\newblock


\bibitem[Liu et~al\mbox{.}(2022)]%
        {liu2022cdarl}
\bibfield{author}{\bibinfo{person}{Zhuang Liu}, \bibinfo{person}{Yunpu Ma}, \bibinfo{person}{Marcel Hildebrandt}, \bibinfo{person}{Yuanxin Ouyang}, {and} \bibinfo{person}{Zhang Xiong}.} \bibinfo{year}{2022}\natexlab{}.
\newblock \showarticletitle{CDARL: a contrastive discriminator-augmented reinforcement learning framework for sequential recommendations}.
\newblock \bibinfo{journal}{\emph{Knowledge and Information Systems}} \bibinfo{volume}{64}, \bibinfo{number}{8} (\bibinfo{year}{2022}), \bibinfo{pages}{2239--2265}.
\newblock


\bibitem[Liu et~al\mbox{.}(2025)]%
        {liu2025pone}
\bibfield{author}{\bibinfo{person}{Ziyang Liu}, \bibinfo{person}{Chaokun Wang}, \bibinfo{person}{Shuwen Zheng}, \bibinfo{person}{Cheng Wu}, \bibinfo{person}{Kai Zheng}, \bibinfo{person}{Yang Song}, {and} \bibinfo{person}{Na Mou}.} \bibinfo{year}{2025}\natexlab{}.
\newblock \showarticletitle{Pone-GNN: Integrating Positive and Negative Feedback in Graph Neural Networks for Recommender Systems}.
\newblock \bibinfo{journal}{\emph{ACM Transactions on Recommender Systems}} (\bibinfo{year}{2025}).
\newblock


\bibitem[Lu et~al\mbox{.}(2021)]%
        {lu2021srecgan}
\bibfield{author}{\bibinfo{person}{Guangben Lu}, \bibinfo{person}{Ziheng Zhao}, \bibinfo{person}{Xiaofeng Gao}, {and} \bibinfo{person}{Guihai Chen}.} \bibinfo{year}{2021}\natexlab{}.
\newblock \showarticletitle{SRecGAN: pairwise adversarial training for sequential recommendation}. In \bibinfo{booktitle}{\emph{International Conference on Database Systems for Advanced Applications}}. Springer, \bibinfo{pages}{20--35}.
\newblock


\bibitem[Lu and Lu(2020)]%
        {lu2020universal}
\bibfield{author}{\bibinfo{person}{Yulong Lu} {and} \bibinfo{person}{Jianfeng Lu}.} \bibinfo{year}{2020}\natexlab{}.
\newblock \showarticletitle{A universal approximation theorem of deep neural networks for expressing probability distributions}.
\newblock \bibinfo{journal}{\emph{Advances in neural information processing systems}}  \bibinfo{volume}{33} (\bibinfo{year}{2020}), \bibinfo{pages}{3094--3105}.
\newblock


\bibitem[Ma et~al\mbox{.}(2024)]%
        {ma2024plug}
\bibfield{author}{\bibinfo{person}{Haokai Ma}, \bibinfo{person}{Ruobing Xie}, \bibinfo{person}{Lei Meng}, \bibinfo{person}{Xin Chen}, \bibinfo{person}{Xu Zhang}, \bibinfo{person}{Leyu Lin}, {and} \bibinfo{person}{Zhanhui Kang}.} \bibinfo{year}{2024}\natexlab{}.
\newblock \showarticletitle{Plug-in diffusion model for sequential recommendation}. In \bibinfo{booktitle}{\emph{Proceedings of the AAAI Conference on Artificial Intelligence}}, Vol.~\bibinfo{volume}{38}. \bibinfo{pages}{8886--8894}.
\newblock


\bibitem[Martin et~al\mbox{.}(2020)]%
        {martin2020monte}
\bibfield{author}{\bibinfo{person}{Alice Martin}, \bibinfo{person}{Charles Ollion}, \bibinfo{person}{Florian Strub}, \bibinfo{person}{Sylvain~Le Corff}, {and} \bibinfo{person}{Olivier Pietquin}.} \bibinfo{year}{2020}\natexlab{}.
\newblock \showarticletitle{The Monte Carlo Transformer: a stochastic self-attention model for sequence prediction}.
\newblock \bibinfo{journal}{\emph{arXiv preprint arXiv:2007.08620}} (\bibinfo{year}{2020}).
\newblock


\bibitem[Noci et~al\mbox{.}(2024)]%
        {noci2024shaped}
\bibfield{author}{\bibinfo{person}{Lorenzo Noci}, \bibinfo{person}{Chuning Li}, \bibinfo{person}{Mufan Li}, \bibinfo{person}{Bobby He}, \bibinfo{person}{Thomas Hofmann}, \bibinfo{person}{Chris~J Maddison}, {and} \bibinfo{person}{Dan Roy}.} \bibinfo{year}{2024}\natexlab{}.
\newblock \showarticletitle{The shaped transformer: Attention models in the infinite depth-and-width limit}.
\newblock \bibinfo{journal}{\emph{Advances in Neural Information Processing Systems}}  \bibinfo{volume}{36} (\bibinfo{year}{2024}).
\newblock


\bibitem[Puthiya~Parambath et~al\mbox{.}(2016)]%
        {puthiya2016coverage}
\bibfield{author}{\bibinfo{person}{Shameem~A Puthiya~Parambath}, \bibinfo{person}{Nicolas Usunier}, {and} \bibinfo{person}{Yves Grandvalet}.} \bibinfo{year}{2016}\natexlab{}.
\newblock \showarticletitle{A coverage-based approach to recommendation diversity on similarity graph}. In \bibinfo{booktitle}{\emph{Proceedings of the 10th ACM Conference on Recommender Systems}}. \bibinfo{pages}{15--22}.
\newblock


\bibitem[Qin et~al\mbox{.}(2024)]%
        {qin2024intent}
\bibfield{author}{\bibinfo{person}{Xiuyuan Qin}, \bibinfo{person}{Huanhuan Yuan}, \bibinfo{person}{Pengpeng Zhao}, \bibinfo{person}{Guanfeng Liu}, \bibinfo{person}{Fuzhen Zhuang}, {and} \bibinfo{person}{Victor~S Sheng}.} \bibinfo{year}{2024}\natexlab{}.
\newblock \showarticletitle{Intent Contrastive Learning with Cross Subsequences for Sequential Recommendation}. In \bibinfo{booktitle}{\emph{Proceedings of the 17th ACM International Conference on Web Search and Data Mining}}. \bibinfo{pages}{548--556}.
\newblock


\bibitem[Qiu et~al\mbox{.}(2022)]%
        {qiu2022contrastive}
\bibfield{author}{\bibinfo{person}{Ruihong Qiu}, \bibinfo{person}{Zi Huang}, \bibinfo{person}{Hongzhi Yin}, {and} \bibinfo{person}{Zijian Wang}.} \bibinfo{year}{2022}\natexlab{}.
\newblock \showarticletitle{Contrastive learning for representation degeneration problem in sequential recommendation}. In \bibinfo{booktitle}{\emph{Proceedings of the fifteenth ACM international conference on web search and data mining}}. \bibinfo{pages}{813--823}.
\newblock


\bibitem[Rashed et~al\mbox{.}(2022)]%
        {rashed2022context}
\bibfield{author}{\bibinfo{person}{Ahmed Rashed}, \bibinfo{person}{Shereen Elsayed}, {and} \bibinfo{person}{Lars Schmidt-Thieme}.} \bibinfo{year}{2022}\natexlab{}.
\newblock \showarticletitle{Context and attribute-aware sequential recommendation via cross-attention}. In \bibinfo{booktitle}{\emph{Proceedings of the 16th ACM Conference on Recommender Systems}}. \bibinfo{pages}{71--80}.
\newblock


\bibitem[Ren et~al\mbox{.}(2020)]%
        {ren2020sequential}
\bibfield{author}{\bibinfo{person}{Ruiyang Ren}, \bibinfo{person}{Zhaoyang Liu}, \bibinfo{person}{Yaliang Li}, \bibinfo{person}{Wayne~Xin Zhao}, \bibinfo{person}{Hui Wang}, \bibinfo{person}{Bolin Ding}, {and} \bibinfo{person}{Ji-Rong Wen}.} \bibinfo{year}{2020}\natexlab{}.
\newblock \showarticletitle{Sequential recommendation with self-attentive multi-adversarial network}. In \bibinfo{booktitle}{\emph{Proceedings of the 43rd international ACM SIGIR conference on research and development in information retrieval}}. \bibinfo{pages}{89--98}.
\newblock


\bibitem[Shenbin et~al\mbox{.}(2020)]%
        {shenbin2020recvae}
\bibfield{author}{\bibinfo{person}{Ilya Shenbin}, \bibinfo{person}{Anton Alekseev}, \bibinfo{person}{Elena Tutubalina}, \bibinfo{person}{Valentin Malykh}, {and} \bibinfo{person}{Sergey~I Nikolenko}.} \bibinfo{year}{2020}\natexlab{}.
\newblock \showarticletitle{Recvae: A new variational autoencoder for top-n recommendations with implicit feedback}. In \bibinfo{booktitle}{\emph{Proceedings of the 13th international conference on web search and data mining}}. \bibinfo{pages}{528--536}.
\newblock


\bibitem[Sun et~al\mbox{.}(2019)]%
        {sun2019bert4rec}
\bibfield{author}{\bibinfo{person}{Fei Sun}, \bibinfo{person}{Jun Liu}, \bibinfo{person}{Jian Wu}, \bibinfo{person}{Changhua Pei}, \bibinfo{person}{Xiao Lin}, \bibinfo{person}{Wenwu Ou}, {and} \bibinfo{person}{Peng Jiang}.} \bibinfo{year}{2019}\natexlab{}.
\newblock \showarticletitle{BERT4Rec: Sequential recommendation with bidirectional encoder representations from transformer}. In \bibinfo{booktitle}{\emph{Proceedings of the 28th ACM international conference on information and knowledge management}}. \bibinfo{pages}{1441--1450}.
\newblock


\bibitem[Tan et~al\mbox{.}(2021)]%
        {tan2021dynamic}
\bibfield{author}{\bibinfo{person}{Qiaoyu Tan}, \bibinfo{person}{Jianwei Zhang}, \bibinfo{person}{Ninghao Liu}, \bibinfo{person}{Xiao Huang}, \bibinfo{person}{Hongxia Yang}, \bibinfo{person}{Jingren Zhou}, {and} \bibinfo{person}{Xia Hu}.} \bibinfo{year}{2021}\natexlab{}.
\newblock \showarticletitle{Dynamic memory based attention network for sequential recommendation}. In \bibinfo{booktitle}{\emph{Proceedings of the AAAI conference on artificial intelligence}}, Vol.~\bibinfo{volume}{35}. \bibinfo{pages}{4384--4392}.
\newblock


\bibitem[Tang and Matteson(2021)]%
        {tang2021probabilistic}
\bibfield{author}{\bibinfo{person}{Binh Tang} {and} \bibinfo{person}{David~S Matteson}.} \bibinfo{year}{2021}\natexlab{}.
\newblock \showarticletitle{Probabilistic transformer for time series analysis}.
\newblock \bibinfo{journal}{\emph{Advances in Neural Information Processing Systems}}  \bibinfo{volume}{34} (\bibinfo{year}{2021}), \bibinfo{pages}{23592--23608}.
\newblock


\bibitem[Tang and Wang(2018)]%
        {tang2018personalized}
\bibfield{author}{\bibinfo{person}{Jiaxi Tang} {and} \bibinfo{person}{Ke Wang}.} \bibinfo{year}{2018}\natexlab{}.
\newblock \showarticletitle{Personalized top-n sequential recommendation via convolutional sequence embedding}. In \bibinfo{booktitle}{\emph{Proceedings of the eleventh ACM international conference on web search and data mining}}. \bibinfo{pages}{565--573}.
\newblock


\bibitem[Tian et~al\mbox{.}(2024)]%
        {tian2024eulerformer}
\bibfield{author}{\bibinfo{person}{Zhen Tian}, \bibinfo{person}{Wayne~Xin Zhao}, \bibinfo{person}{Changwang Zhang}, \bibinfo{person}{Xin Zhao}, \bibinfo{person}{Zhongrui Ma}, {and} \bibinfo{person}{Ji-Rong Wen}.} \bibinfo{year}{2024}\natexlab{}.
\newblock \showarticletitle{EulerFormer: Sequential User Behavior Modeling with Complex Vector Attention}. In \bibinfo{booktitle}{\emph{Proceedings of the 47th International ACM SIGIR Conference on Research and Development in Information Retrieval}}. \bibinfo{pages}{1619--1628}.
\newblock


\bibitem[Tran et~al\mbox{.}(2023)]%
        {tran2023attention}
\bibfield{author}{\bibinfo{person}{Viet~Anh Tran}, \bibinfo{person}{Guillaume Salha-Galvan}, \bibinfo{person}{Bruno Sguerra}, {and} \bibinfo{person}{Romain Hennequin}.} \bibinfo{year}{2023}\natexlab{}.
\newblock \showarticletitle{Attention mixtures for time-aware sequential recommendation}. In \bibinfo{booktitle}{\emph{Proceedings of the 46th International ACM SIGIR Conference on Research and Development in Information Retrieval}}. \bibinfo{pages}{1821--1826}.
\newblock


\bibitem[Vaswani et~al\mbox{.}(2017)]%
        {vaswani2017attention}
\bibfield{author}{\bibinfo{person}{Ashish Vaswani}, \bibinfo{person}{Noam Shazeer}, \bibinfo{person}{Niki Parmar}, \bibinfo{person}{Jakob Uszkoreit}, \bibinfo{person}{Llion Jones}, \bibinfo{person}{Aidan~N Gomez}, \bibinfo{person}{{\L}ukasz Kaiser}, {and} \bibinfo{person}{Illia Polosukhin}.} \bibinfo{year}{2017}\natexlab{}.
\newblock \showarticletitle{Attention is all you need}.
\newblock \bibinfo{journal}{\emph{Advances in neural information processing systems}}  \bibinfo{volume}{30} (\bibinfo{year}{2017}).
\newblock


\bibitem[Wang et~al\mbox{.}(2023a)]%
        {wang2023sequential}
\bibfield{author}{\bibinfo{person}{Chenyang Wang}, \bibinfo{person}{Weizhi Ma}, \bibinfo{person}{Chong Chen}, \bibinfo{person}{Min Zhang}, \bibinfo{person}{Yiqun Liu}, {and} \bibinfo{person}{Shaoping Ma}.} \bibinfo{year}{2023}\natexlab{a}.
\newblock \showarticletitle{Sequential recommendation with multiple contrast signals}.
\newblock \bibinfo{journal}{\emph{ACM Transactions on Information Systems (TOIS)}} \bibinfo{volume}{41}, \bibinfo{number}{1} (\bibinfo{year}{2023}), \bibinfo{pages}{1--27}.
\newblock


\bibitem[Wang et~al\mbox{.}(2017)]%
        {wang2017irgan}
\bibfield{author}{\bibinfo{person}{Jun Wang}, \bibinfo{person}{Lantao Yu}, \bibinfo{person}{Weinan Zhang}, \bibinfo{person}{Yu Gong}, \bibinfo{person}{Yinghui Xu}, \bibinfo{person}{Benyou Wang}, \bibinfo{person}{Peng Zhang}, {and} \bibinfo{person}{Dell Zhang}.} \bibinfo{year}{2017}\natexlab{}.
\newblock \showarticletitle{Irgan: A minimax game for unifying generative and discriminative information retrieval models}. In \bibinfo{booktitle}{\emph{Proceedings of the 40th International ACM SIGIR conference on Research and Development in Information Retrieval}}. \bibinfo{pages}{515--524}.
\newblock


\bibitem[Wang et~al\mbox{.}(2023b)]%
        {wang2023better}
\bibfield{author}{\bibinfo{person}{Zekai Wang}, \bibinfo{person}{Tianyu Pang}, \bibinfo{person}{Chao Du}, \bibinfo{person}{Min Lin}, \bibinfo{person}{Weiwei Liu}, {and} \bibinfo{person}{Shuicheng Yan}.} \bibinfo{year}{2023}\natexlab{b}.
\newblock \showarticletitle{Better diffusion models further improve adversarial training}. In \bibinfo{booktitle}{\emph{International Conference on Machine Learning}}. PMLR, \bibinfo{pages}{36246--36263}.
\newblock


\bibitem[Wu et~al\mbox{.}(2019b)]%
        {wu2019neural}
\bibfield{author}{\bibinfo{person}{Le Wu}, \bibinfo{person}{Peijie Sun}, \bibinfo{person}{Yanjie Fu}, \bibinfo{person}{Richang Hong}, \bibinfo{person}{Xiting Wang}, {and} \bibinfo{person}{Meng Wang}.} \bibinfo{year}{2019}\natexlab{b}.
\newblock \showarticletitle{A neural influence diffusion model for social recommendation}. In \bibinfo{booktitle}{\emph{Proceedings of the 42nd international ACM SIGIR conference on research and development in information retrieval}}. \bibinfo{pages}{235--244}.
\newblock


\bibitem[Wu et~al\mbox{.}(2019a)]%
        {wu2019pd}
\bibfield{author}{\bibinfo{person}{Qiong Wu}, \bibinfo{person}{Yong Liu}, \bibinfo{person}{Chunyan Miao}, \bibinfo{person}{Binqiang Zhao}, \bibinfo{person}{Yin Zhao}, {and} \bibinfo{person}{Lu Guan}.} \bibinfo{year}{2019}\natexlab{a}.
\newblock \showarticletitle{PD-GAN: Adversarial Learning for Personalized Diversity-Promoting Recommendation.}. In \bibinfo{booktitle}{\emph{IJCAI}}, Vol.~\bibinfo{volume}{19}. \bibinfo{pages}{3870--3876}.
\newblock


\bibitem[Xie et~al\mbox{.}(2022)]%
        {xie2022contrastive}
\bibfield{author}{\bibinfo{person}{Xu Xie}, \bibinfo{person}{Fei Sun}, \bibinfo{person}{Zhaoyang Liu}, \bibinfo{person}{Shiwen Wu}, \bibinfo{person}{Jinyang Gao}, \bibinfo{person}{Jiandong Zhang}, \bibinfo{person}{Bolin Ding}, {and} \bibinfo{person}{Bin Cui}.} \bibinfo{year}{2022}\natexlab{}.
\newblock \showarticletitle{Contrastive learning for sequential recommendation}. In \bibinfo{booktitle}{\emph{2022 IEEE 38th international conference on data engineering (ICDE)}}. IEEE, \bibinfo{pages}{1259--1273}.
\newblock


\bibitem[Xie et~al\mbox{.}(2021)]%
        {xie2021adversarial}
\bibfield{author}{\bibinfo{person}{Zhe Xie}, \bibinfo{person}{Chengxuan Liu}, \bibinfo{person}{Yichi Zhang}, \bibinfo{person}{Hongtao Lu}, \bibinfo{person}{Dong Wang}, {and} \bibinfo{person}{Yue Ding}.} \bibinfo{year}{2021}\natexlab{}.
\newblock \showarticletitle{Adversarial and contrastive variational autoencoder for sequential recommendation}. In \bibinfo{booktitle}{\emph{Proceedings of the web conference 2021}}. \bibinfo{pages}{449--459}.
\newblock


\bibitem[Xu et~al\mbox{.}(2019a)]%
        {xu2019graph}
\bibfield{author}{\bibinfo{person}{Chengfeng Xu}, \bibinfo{person}{Pengpeng Zhao}, \bibinfo{person}{Yanchi Liu}, \bibinfo{person}{Victor~S Sheng}, \bibinfo{person}{Jiajie Xu}, \bibinfo{person}{Fuzhen Zhuang}, \bibinfo{person}{Junhua Fang}, {and} \bibinfo{person}{Xiaofang Zhou}.} \bibinfo{year}{2019}\natexlab{a}.
\newblock \showarticletitle{Graph contextualized self-attention network for session-based recommendation.}. In \bibinfo{booktitle}{\emph{IJCAI}}, Vol.~\bibinfo{volume}{19}. \bibinfo{pages}{3940--3946}.
\newblock


\bibitem[Xu et~al\mbox{.}(2019b)]%
        {xu2019recurrent}
\bibfield{author}{\bibinfo{person}{Chengfeng Xu}, \bibinfo{person}{Pengpeng Zhao}, \bibinfo{person}{Yanchi Liu}, \bibinfo{person}{Jiajie Xu}, \bibinfo{person}{Victor S~Sheng S.~Sheng}, \bibinfo{person}{Zhiming Cui}, \bibinfo{person}{Xiaofang Zhou}, {and} \bibinfo{person}{Hui Xiong}.} \bibinfo{year}{2019}\natexlab{b}.
\newblock \showarticletitle{Recurrent convolutional neural network for sequential recommendation}. In \bibinfo{booktitle}{\emph{The world wide web conference}}. \bibinfo{pages}{3398--3404}.
\newblock


\bibitem[Xu et~al\mbox{.}(2024)]%
        {xu2024online}
\bibfield{author}{\bibinfo{person}{Jianyu Xu}, \bibinfo{person}{Bin Liu}, \bibinfo{person}{Xiujie Zhao}, {and} \bibinfo{person}{Xiao-Lin Wang}.} \bibinfo{year}{2024}\natexlab{}.
\newblock \showarticletitle{Online reinforcement learning for condition-based group maintenance using factored Markov decision processes}.
\newblock \bibinfo{journal}{\emph{European Journal of Operational Research}} \bibinfo{volume}{315}, \bibinfo{number}{1} (\bibinfo{year}{2024}), \bibinfo{pages}{176--190}.
\newblock


\bibitem[Yan et~al\mbox{.}(2019)]%
        {yan2019cosrec}
\bibfield{author}{\bibinfo{person}{An Yan}, \bibinfo{person}{Shuo Cheng}, \bibinfo{person}{Wang-Cheng Kang}, \bibinfo{person}{Mengting Wan}, {and} \bibinfo{person}{Julian McAuley}.} \bibinfo{year}{2019}\natexlab{}.
\newblock \showarticletitle{CosRec: 2D convolutional neural networks for sequential recommendation}. In \bibinfo{booktitle}{\emph{Proceedings of the 28th ACM international conference on information and knowledge management}}. \bibinfo{pages}{2173--2176}.
\newblock


\bibitem[Ying et~al\mbox{.}(2018)]%
        {ying2018sequential}
\bibfield{author}{\bibinfo{person}{Haochao Ying}, \bibinfo{person}{Fuzhen Zhuang}, \bibinfo{person}{Fuzheng Zhang}, \bibinfo{person}{Yanchi Liu}, \bibinfo{person}{Guandong Xu}, \bibinfo{person}{Xing Xie}, \bibinfo{person}{Hui Xiong}, {and} \bibinfo{person}{Jian Wu}.} \bibinfo{year}{2018}\natexlab{}.
\newblock \showarticletitle{Sequential recommender system based on hierarchical attention network}. In \bibinfo{booktitle}{\emph{IJCAI international joint conference on artificial intelligence}}.
\newblock


\bibitem[Yoon et~al\mbox{.}(2017)]%
        {yoon2017semi}
\bibfield{author}{\bibinfo{person}{Andre~S Yoon}, \bibinfo{person}{Taehoon Lee}, \bibinfo{person}{Yongsub Lim}, \bibinfo{person}{Deokwoo Jung}, \bibinfo{person}{Philgyun Kang}, \bibinfo{person}{Dongwon Kim}, \bibinfo{person}{Keuntae Park}, {and} \bibinfo{person}{Yongjin Choi}.} \bibinfo{year}{2017}\natexlab{}.
\newblock \showarticletitle{Semi-supervised learning with deep generative models for asset failure prediction}.
\newblock \bibinfo{journal}{\emph{arXiv preprint arXiv:1709.00845}} (\bibinfo{year}{2017}).
\newblock


\bibitem[Yu et~al\mbox{.}(2022)]%
        {yu2022element}
\bibfield{author}{\bibinfo{person}{Le Yu}, \bibinfo{person}{Guanghui Wu}, \bibinfo{person}{Leilei Sun}, \bibinfo{person}{Bowen Du}, {and} \bibinfo{person}{Weifeng Lv}.} \bibinfo{year}{2022}\natexlab{}.
\newblock \showarticletitle{Element-guided Temporal Graph Representation Learning for Temporal Sets Prediction}. In \bibinfo{booktitle}{\emph{Proceedings of the ACM Web Conference 2022}}. \bibinfo{pages}{1902--1913}.
\newblock


\bibitem[Yuan et~al\mbox{.}(2020)]%
        {yuan2020attention}
\bibfield{author}{\bibinfo{person}{Weihua Yuan}, \bibinfo{person}{Hong Wang}, \bibinfo{person}{Xiaomei Yu}, \bibinfo{person}{Nan Liu}, {and} \bibinfo{person}{Zhenghao Li}.} \bibinfo{year}{2020}\natexlab{}.
\newblock \showarticletitle{Attention-based context-aware sequential recommendation model}.
\newblock \bibinfo{journal}{\emph{Information Sciences}}  \bibinfo{volume}{510} (\bibinfo{year}{2020}), \bibinfo{pages}{122--134}.
\newblock


\bibitem[Zhang and Hurley(2008)]%
        {zhang2008avoiding}
\bibfield{author}{\bibinfo{person}{Mi Zhang} {and} \bibinfo{person}{Neil Hurley}.} \bibinfo{year}{2008}\natexlab{}.
\newblock \showarticletitle{Avoiding monotony: improving the diversity of recommendation lists}. In \bibinfo{booktitle}{\emph{Proceedings of the 2008 ACM conference on Recommender systems}}. \bibinfo{pages}{123--130}.
\newblock


\bibitem[Zhang et~al\mbox{.}(2022)]%
        {zhang2022dynamic}
\bibfield{author}{\bibinfo{person}{Mengqi Zhang}, \bibinfo{person}{Shu Wu}, \bibinfo{person}{Xueli Yu}, \bibinfo{person}{Qiang Liu}, {and} \bibinfo{person}{Liang Wang}.} \bibinfo{year}{2022}\natexlab{}.
\newblock \showarticletitle{Dynamic graph neural networks for sequential recommendation}.
\newblock \bibinfo{journal}{\emph{IEEE Transactions on Knowledge and Data Engineering}} \bibinfo{volume}{35}, \bibinfo{number}{5} (\bibinfo{year}{2022}), \bibinfo{pages}{4741--4753}.
\newblock


\bibitem[Zhang et~al\mbox{.}(2019)]%
        {zhang2019feature}
\bibfield{author}{\bibinfo{person}{Tingting Zhang}, \bibinfo{person}{Pengpeng Zhao}, \bibinfo{person}{Yanchi Liu}, \bibinfo{person}{Victor~S Sheng}, \bibinfo{person}{Jiajie Xu}, \bibinfo{person}{Deqing Wang}, \bibinfo{person}{Guanfeng Liu}, \bibinfo{person}{Xiaofang Zhou}, {et~al\mbox{.}}} \bibinfo{year}{2019}\natexlab{}.
\newblock \showarticletitle{Feature-level deeper self-attention network for sequential recommendation.}. In \bibinfo{booktitle}{\emph{IJCAI}}. \bibinfo{pages}{4320--4326}.
\newblock


\bibitem[Zhao et~al\mbox{.}(2024)]%
        {zhao2024denoising}
\bibfield{author}{\bibinfo{person}{Jujia Zhao}, \bibinfo{person}{Wang Wenjie}, \bibinfo{person}{Yiyan Xu}, \bibinfo{person}{Teng Sun}, \bibinfo{person}{Fuli Feng}, {and} \bibinfo{person}{Tat-Seng Chua}.} \bibinfo{year}{2024}\natexlab{}.
\newblock \showarticletitle{Denoising diffusion recommender model}. In \bibinfo{booktitle}{\emph{Proceedings of the 47th International ACM SIGIR Conference on Research and Development in Information Retrieval}}. \bibinfo{pages}{1370--1379}.
\newblock


\bibitem[Zhao et~al\mbox{.}(2021)]%
        {zhao2021variational}
\bibfield{author}{\bibinfo{person}{Jing Zhao}, \bibinfo{person}{Pengpeng Zhao}, \bibinfo{person}{Lei Zhao}, \bibinfo{person}{Yanchi Liu}, \bibinfo{person}{Victor~S Sheng}, {and} \bibinfo{person}{Xiaofang Zhou}.} \bibinfo{year}{2021}\natexlab{}.
\newblock \showarticletitle{Variational self-attention network for sequential recommendation}. In \bibinfo{booktitle}{\emph{2021 IEEE 37th International Conference on Data Engineering (ICDE)}}. IEEE, \bibinfo{pages}{1559--1570}.
\newblock


\bibitem[Zheng et~al\mbox{.}(2020)]%
        {zheng2020modeling}
\bibfield{author}{\bibinfo{person}{Xiaolin Zheng}, \bibinfo{person}{Menghan Wang}, \bibinfo{person}{Renjun Xu}, \bibinfo{person}{Jianmeng Li}, {and} \bibinfo{person}{Yan Wang}.} \bibinfo{year}{2020}\natexlab{}.
\newblock \showarticletitle{Modeling dynamic missingness of implicit feedback for sequential recommendation}.
\newblock \bibinfo{journal}{\emph{IEEE Transactions on Knowledge and Data Engineering}} \bibinfo{volume}{34}, \bibinfo{number}{1} (\bibinfo{year}{2020}), \bibinfo{pages}{405--418}.
\newblock


\bibitem[Zhou et~al\mbox{.}(2024)]%
        {zhou2024theoretical}
\bibfield{author}{\bibinfo{person}{Cai Zhou}, \bibinfo{person}{Rose Yu}, {and} \bibinfo{person}{Yusu Wang}.} \bibinfo{year}{2024}\natexlab{}.
\newblock \showarticletitle{On the Theoretical Expressive Power and the Design Space of Higher-Order Graph Transformers}. In \bibinfo{booktitle}{\emph{International Conference on Artificial Intelligence and Statistics}}. PMLR, \bibinfo{pages}{2179--2187}.
\newblock


\bibitem[Zhou et~al\mbox{.}(2020)]%
        {zhou2020cnn}
\bibfield{author}{\bibinfo{person}{Xiaokang Zhou}, \bibinfo{person}{Yue Li}, {and} \bibinfo{person}{Wei Liang}.} \bibinfo{year}{2020}\natexlab{}.
\newblock \showarticletitle{CNN-RNN based intelligent recommendation for online medical pre-diagnosis support}.
\newblock \bibinfo{journal}{\emph{IEEE/ACM Transactions on Computational Biology and Bioinformatics}} \bibinfo{volume}{18}, \bibinfo{number}{3} (\bibinfo{year}{2020}), \bibinfo{pages}{912--921}.
\newblock


\bibitem[Zivic et~al\mbox{.}(2024)]%
        {zivic2024scaling}
\bibfield{author}{\bibinfo{person}{Pablo Zivic}, \bibinfo{person}{Hernan Vazquez}, {and} \bibinfo{person}{Jorge S{\'a}nchez}.} \bibinfo{year}{2024}\natexlab{}.
\newblock \showarticletitle{Scaling Sequential Recommendation Models with Transformers}. In \bibinfo{booktitle}{\emph{Proceedings of the 47th International ACM SIGIR Conference on Research and Development in Information Retrieval}}. \bibinfo{pages}{1567--1577}.
\newblock


\bibitem[Zolghadr et~al\mbox{.}(2024)]%
        {zolghadr2024generative}
\bibfield{author}{\bibinfo{person}{Sharare Zolghadr}, \bibinfo{person}{Ole Winther}, {and} \bibinfo{person}{Paul Jeha}.} \bibinfo{year}{2024}\natexlab{}.
\newblock \showarticletitle{Generative Diffusion Models for Sequential Recommendations}.
\newblock \bibinfo{journal}{\emph{arXiv preprint arXiv:2410.19429}} (\bibinfo{year}{2024}).
\newblock


\end{thebibliography}


\end{document}